\documentclass[a4paper,USenglish,cleveref,autoref,thm-restate]{lipics-v2021}
\usepackage{auxlib}

\usepackage{float}
\usepackage{xspace}
\usepackage{calc}
\usepackage{onlyamsmath}
\createcustomnote{ftd}{backgroundcolor=AUXLlightblue, prefix=Felix}
\createcustomnote{ctodo}{backgroundcolor=orange!50, prefix=Chr.:}

\makeatletter
\newcommand{\DeclareMathActive}[2]{%
  \expandafter\edef\csname keep@#1@code\endcsname{\mathchar\the\mathcode`#1 }
  \begingroup\lccode`~=`#1\relax
  \lowercase{\endgroup\def~}{#2}%
  \AtBeginDocument{\mathcode`#1="8000 }%
}

\newcommand{\std}[1]{\csname keep@#1@code\endcsname}
\patchcmd{\newmcodes@}{\mathcode`\-\relax}{\std@minuscode\relax}{}{\ddt}
\AtBeginDocument{\edef\std@minuscode{\the\mathcode`-}}
\makeatother

\makeatletter
\def\Igbl_#1{\std{I}_{\text{\upshape #1}}}
\def\Ispl{\@ifnextchar_{\Igbl}{\std{I}}}
\DeclareMathActive{I}{\Ispl}
\DeclareMathActive{O}{\operatorname{\std{O}}}
\makeatother

\crefname{enumi}{Statement}{Statements} %change cref's label "Item" to "Statement"
\crefname{rule}{Rule}{Rules} % change cref's label 

\makeatletter
\def\HyPsd@Warning#1{}
\makeatother

\makeatletter
\def\?#1{}
\def\whp{w.h.p\@ifnextchar.{.\whpFootnote/\?}{\@ifnextchar,{.\whpFootnote/}{\@ifnextchar){.\whpFootnote/}{.\whpFootnote/\ }}}}
\def\Whp{W.h.p\@ifnextchar.{.\whpFootnote/\?}{\@ifnextchar,{\whpFootnote/.}{.\whpFootnote/\ }}}
\makeatother

\usepackage[ruled,noend,noline,nofillcomment,linesnumbered]{algorithm2e}

\crefname{algocfline}{Algorithm}{Algorithms}
\DontPrintSemicolon
\SetCommentSty{mycommfont}

\newcommand{\Clock}[1]{\mathsf{clock}[#1]}

\newcommand{\Subphase}[1]{\mathsf{subphase}[#1]}

\newcommand{\Input}[1]{\mathsf{input}[#1]}
\newcommand{\Work}[1]{\mathsf{work}[#1]}
\newcommand{\Opinion}[1]{\mathsf{opinion}[#1]}
\newcommand{\Output}[1]{\mathsf{output}[#1]}
\newcommand{\stage}[1]{\mathbb{T}_{#1}}

\newcommand{\Stop}{\mathsf{stop}}
\newcommand{\FewInGather}{A}
\newcommand{\NoBroadcast}{B}

\newcommand*{\polyaConstSmall}{\ensuremath{\epsilon_p}}
\def\PE{\operatorname{PE}}

\makeatletter

\renewcommand\paragraph{\@startsection{subparagraph}{5}{\z@}%
                                       {1.625ex \@plus0.5ex \@minus .1ex}%
                                       {-1em}%
                                      {\sffamily\normalsize\bfseries}}
\let\dcmparagraph\paragraph
\def\paragraph#1{\dcmparagraph{#1.}}

\makeatother

\usepackage[subtle]{savetrees}

\title{\mbox{Loosely-Stabilizing Phase Clocks and the} \mbox{Adaptive Majority Problem}}
\author{Petra Berenbrink}{Universität Hamburg, Germany}{petra.berenbrink@uni-hamburg.de}{}{}

\author{Felix Biermeier}{Universität Hamburg, Germany}{felix.biermeier@uni-hamburg.de}{}{}

\author{Christopher Hahn}{Universität Hamburg, Germany}{tim.christopher.hahn@uni-hamburg.de}{}{}

\author{Dominik Kaaser}{Universität Hamburg, Germany}{dominik.kaaser@uni-hamburg.de}{https://orcid.org/0000-0002-2083-7145}{}

\authorrunning{P.\ Berenbrink, F.\ Biermeier, C.\ Hahn, D.\ Kaaser}

\Copyright{Petra Berenbrink, Felix Biermeier, Christopher Hahn, and Dominik Kaaser}

 \nolinenumbers %uncomment to disable line numbering
\modulolinenumbers

\hideLIPIcs  %uncomment to remove references to LIPIcs series (logo, DOI, ...), e.g. when preparing a pre-final version to be uploaded to arXiv or another public repository

\let\epsilon\varepsilon
\begin{document}

\maketitle
\begin{abstract}
We present a loosely-stabilizing phase clock for population protocols.
In the population model we are given a system of $n$ identical agents which interact in a sequence of randomly chosen pairs.
Our phase clock is leaderless and it requires $O(\log n)$ states.
It runs forever and is, at any point of time, in a synchronous state w.h.p.
When started in an arbitrary configuration, it recovers rapidly and enters a synchronous configuration within $O(n\log n)$ interactions w.h.p.
Once the clock is synchronized, it stays in a synchronous configuration for at least $\poly n$ parallel time w.h.p.

We use our clock to design a loosely-stabilizing protocol that  solves the comparison problem introduced by Alistarh et al., 2021.
In this problem, a subset of agents has at any time either $A$ or $B$ as input.
The goal is to keep track which of the two opinions is (momentarily) the majority.
We show that if the majority has a support of at least $\Omega(\log n)$ agents and a sufficiently large bias is present, then the protocol converges to a correct output within $O(n\log n)$ interactions and stays in a correct configuration for $\poly n$ interactions, w.h.p.
\end{abstract}

\section{Introduction}
In this paper we introduce a loosely-stabilizing leaderless phase clock for the population model and demonstrate its usability by applying the clock to the comparison problem introduced in \cite{DBLP:conf/podc/Alistarh0U21}. Population protocols have been introduced by Angluin et al.~\cite{DBLP:journals/dc/AngluinADFP06}. A population consists of $n$ anonymous agents.
A random scheduler selects in discrete time steps pairs of agents to interact.
The interacting agents execute a state transition, as specified by the \emph{algorithm} of the population protocol. 
Angluin et al.~\cite{DBLP:journals/dc/AngluinADFP06} gave a variety of motivating examples for the population model,
including averaging in sensor networks, or modeling a disease monitoring system for a flock of birds.
In \cite{DBLP:journals/tcs/SudoNYOKM12} the authors introduce the notion of loose-stabilization.
A population protocol is loosely-stabilizing if, from an arbitrary state, it reaches a state with correct output fast and remains in such a state for a polynomial number of interactions.
In contrast, self-stabilizing protocols are required to converge  to the correct output state from any possible initial configuration and stay in a correct configuration indefinitely.
Many population protocols heavily rely on so-called \emph{phase clocks}  which divide the interactions into blocks of $\LDAUOmicron{n \log n}$ interactions each. The phase clocks are used to synchronize population protocols. 
For example, in \cite{DBLP:conf/soda/GasieniecS18, DBLP:conf/stoc/BerenbrinkGK20} they are used to efficiently solve leader election and in \cite{DBLP:journals/dc/BerenbrinkEFKKR21} they are used to solve the majority problem.

\medskip

In the first part of this paper we present a loosely-stabilizing and leaderless phase clock with $\LDAUOmicron{\log n}$ many states per agent.
We show that this clock can run forever and that, at any point of time, it is synchronized \whp.
In contrast to related work \cite{DBLP:conf/soda/AlistarhAG18,DBLP:journals/dc/AngluinAE08a,DBLP:journals/dc/BerenbrinkEFKKR21,DBLP:conf/soda/GasieniecS18}, our clock protocol \emph{recovers} rapidly in case of an error: from an arbitrary configuration it always enters a synchronous configuration within $\LDAUOmicron{n\log{n}}$ interactions \whp. Once synchronized it stays in a synchronous configuration for at least $\poly{n}$ interactions, \whp.
Our phase clock can be used to synchronize population protocols into phases of $\LDAUOmicron{n \log n}$ interactions, guaranteeing that there is a big overlap between the phases of any pair of agents. Our clock protocol is simple, robust and easy to use.

In the second part of this paper we demonstrate how to apply our phase clock by solving an \emph{adaptive majority} problem motivated by the work of \cite{DBLP:conf/podc/Alistarh0U21,DBLP:conf/opodis/AmirAL20}. Our problem is defined as follows. Each agent has either opinion $A$, $B$, or $U$ for being neutral.
We say that agents change their input with rate $r$ if in every time step an arbitrary agent can change its opinion with probability $r$.
The goal is to output, at any time, the actual majority opinion.
The idea of our approach is as follows.
Our protocol simply starts, at the beginning of each phase, a \emph{static} majority protocol as a black box. This protocol is takes as an input the set of opinions at that time and calculates the majority opinion over these inputs.  
The outcome of the protocol is then used during the whole next phase as majority opinion.
In order to highlight the simplicity of our phase clock, we first use the very natural protocol based solely on canceling opposing opinions introduced in \cite{DBLP:journals/dc/AngluinAE08a}. 
Then we present a variant based on the undecided state dynamics from \cite{DBLP:journals/dc/AngluinAE08} which works as follows.
The agents have one of two opinions $A$ or $B$, or they are undecided.
Whenever two agents with the same opinion interact, nothing happens.
When two agents with an opposite opinion interact they will become undecided.
Undecided agents interacting with an agent with either opinion $A$ or opinion $B$ adopt that opinion.

%Our protocol simply starts the undecided state dynamics at the beginning of each phase and outputs the result of the previous phase as majority opinion.
%It runs for an arbitrary time and it is $(O(n \log n), \poly n)$-loosely-stabilizing.
%We show that the protocol outputs the correct majority provided a sufficiently large bias is present.
%Additionally, it can handle opinion changes (from $A$ to $B$ and vice versa) at a certain rate $r$ per interaction.
Without loss of generality we assume that $A$ is the majority opinion in the following. 
When at least $\Omega(\log n)$ agents have opinion $A$, there is a constant factor bias between $A$ and $B$, and the opinions change at most at rate $1/n$ per interaction, the system outputs $A$ \whp.
Our protocol requires only $\LDAUOmicron{\log{n}}$ many states.
%This improves upon the $\log{n}\cdot\log\log{n}$ many states required by \cite{DBLP:conf/podc/Alistarh0U21}.
For the setting where all agents have either opinion $A$ or $B$ (none of the agent is in the neutral state $U$) and we have an additive bias of $n^{3/4 + \epsilon}$ for some constant $\epsilon > 0$ is present, the system again converges to $A$ \whp.
In the latter setting we can tolerate a rate of order $r = \ldauOmega{n^{-1/4 + \epsilon}}$.

%At every point of time each agent \whp outputs the correct majority opinion as long as the majority opinion has a support of at least $\Omega(\log n)$ and its support is by a constant fraction larger than the one of the minority opinion.
%For a detailed description of these results and also some extensions see \cref{sec:majority}.

\paragraph{Related Work}
Population protocols have been introduced by Angluin et al.~\cite{DBLP:journals/dc/AngluinADFP06}. 
Many of the early results focus on characterizing the class of problems which are solvable in the population model.
For example, population protocols with a constant number of states can exactly compute predicates which are definable in Presburger arithmetic \cite{DBLP:journals/dc/AngluinADFP06, DBLP:conf/podc/AngluinAE06, DBLP:journals/dc/AngluinAER07}.
There are many results for majority and leader election, see \cite{doty2021time} and \cite{DBLP:conf/stoc/BerenbrinkGK20} for the latest results.
In \cite{DBLP:journals/tcs/SudoNYOKM12} the authors introduce the notion of \emph{loose}-stabilization to mitigate the fact that self-stabilizing protocols usually require some global knowledge on the population size (or a large amount of states).
See \cite{DBLP:conf/podc/BurmanCCDNSX21} for an overview of self-stabilizing population protocols.

In \cite{DBLP:journals/dc/AngluinAE08a} the authors present and analyze a phase clocks which divides the time into phases of $\LDAUOmicron{n \log n}$ interactions assuming that a unique leader exists. They also present a generalization of the clock using a junta of size $n^{\epsilon}$ (for constant $\epsilon$) instead of a unique leader and analyze the process empirically.  
In \cite{DBLP:conf/soda/GasieniecS18} the authors show that the junta-driven phase clock needs $\LDAUOmicron{\log\log n}$ many states and it ticks for a polynomial number of interactions.
The protocol can easily be modified such that it requires only a constant number of states after the junta election \cite{DBLP:journals/dc/BerenbrinkEFKKR21}.
In the brief announcement \cite{DBLP:conf/podc/KosowskiU18} the authors suggest a phase clocks which, similarly to \cite{DBLP:conf/soda/GasieniecS18}, relies on a junta of size at most $n^{\epsilon}$.
Their clocks are based on the oscillatory dynamics from \cite{DBLP:conf/stoc/DudekK18} and need constantly many states in the case that the junta is already selected.
In \cite{DBLP:conf/soda/AlistarhAG18} the authors present a leaderless phase clock with $\LDAUOmicron{\log n}$ states.
In contrast to our leaderless phase clock, the clock from \cite{DBLP:conf/soda/AlistarhAG18} is not self-stabilizing: it runs only for a polynomial number of interactions.
The analysis is based on the potential function analysis introduced in \cite{DBLP:conf/icalp/TalwarW14} for the greedy balls-into-bins strategy where each ball has to be allocated into one out of two randomly chosen bins.
This analysis assumes an initially balanced configuration and it cannot  be adopted to an arbitrary unbalanced state, which would be required to deal with unsynchronized clock configurations. 
In \cite{DBLP:journals/jcss/Aspnes21} the authors consider a variant of the population model, so-called clocked population protocols, where agents have an additional flag for clock ticks.
The clock signal indicates when the agents have waited sufficiently long for a protocol to have converged.
They show that a clocked population protocol running in less than $\omega^k$ time for fixed $k\geq 2$ is equivalent in power to nondeterministic Turing machines with logarithmic space. 

Another line of related work considers the problem of \emph{exact majority}, where one seeks to achieve (guaranteed) majority consensus, even if the additive bias is as small as one~\cite{ DBLP:journals/dc/DotyS18,  DBLP:conf/soda/AlistarhAG18,    DBLP:journals/corr/abs-1805-04586, BKKP20}. The currently best protocol \cite{doty2021time} solves exact majority with $\ldauOmicron{\log n}$ states and $\ldauOmicron{\log n}$ stabilization time, both in expectation and \whp.
The authors of \cite{DBLP:journals/dc/AngluinAE08} solve the approximate majority problem. They introduce the undecided state dynamics in the population model and consider two opinions. They show that their $3$-state protocol reaches consensus \whp in $\ldauOmicron{n \log{n}}$ interactions.
If the bias is of order $\ldauomega{\sqrt n \cdot \log n}$ the undecided state dynamics converges towards the initial majority \whp.
In \cite{DBLP:journals/nc/CondonHKM20} this required bias is reduced to $\Omega(\sqrt{n \log n})$.

%In \cite{DBLP:conf/soda/AlistarhAEGR17} the authors show that any stable majority protocol using $\log\log{n} / 2$ states requires $\ldauOmega{n^2 / {\log\log n}}$ interactions in expectation.
%Assuming some natural properties, \cite{DBLP:conf/soda/AlistarhAG18} show that any majority protocol which stabilizes in expected
%${n^{2-\ldauOmega{1}}}$ interactions requires $\ldauOmega{\log{n}}$ states.
%Note that all of these protocols require that all agents start at time $0$ in a well-defined initial configuration, and in contrast to our work they cannot handle input changes over time.

In \cite{DBLP:conf/dna/AlistarhDKSU17}
the authors develop an algorithm to detect whether there is an agent in a given state $Y$ or not.
They introduce so-called \emph{leak} transitions and \emph{catalyst} transitions.
A catalyst transition for a state $X$ is a transition which does not change the number of agents in state $X$.
Leak transitions are spurious reactions which can consume and create arbitrary non-catalytic agents.
%It is assumed that $Y$ can only participate in catalyst reactions.
%The authors present an algorithm which guarantees the following. Assume that the rate at which leaks occur is upper bounded by $\beta/n\ll 1/n$.
%The algorithm outputs the state of a randomly chosen agent after $\LDAUOmicron{n \log n}$ transitions.
%The authors show that the probability of a false negative is at most $1/e + \LDAUomicron{1}$, and the probability of a false positive is at most $\beta$.
%In \cite{DBLP:conf/stoc/DudekK18} the authors develop population protocols for broadcasting and source detection (the agents are required to decide if at least one agent in a dedicated source state is present in the population).
%Both protocols are based on an oscillatory dynamics.
In \cite{DBLP:conf/podc/Alistarh0U21} the authors introduce the \emph{robust comparison problem} where the goal is to decide which of the two  states $A$ and $B$ have the larger support.
The authors adopt the model of \cite{DBLP:conf/dna/AlistarhDKSU17} with leak transitions and catalyst transitions.
For the case that the initial support of $A$ and $B$ is at least $\LDAUOmicron{\log n}$ the authors present a loosely-stabilizing dynamics.
If at least $\Omega(\log n)$ agents are in either $A$ or $B$ and the ratio between the numbers of agents supporting $A$ and $B$ is at least a constant, their protocol solves the problem with $\LDAUOmicron{\log n\cdot\log\log{n}}$ states per agent.
It converges in $\LDAUOmicron{n\log n}$ interactions such that every agent outputs the majority\whp.
If the initial support of $A$ and $B$ states is $\LDAUOmega{\log^2 n}$ the authors can strengthen their results such that a ratio between the two base states of $1+\LDAUomicron{1}$ is sufficient.
The results also hold with leak transitions not affecting agents in state $A$ or $B$ with rate $1/n$.
In this case the authors show that most of the agents output the correct majority.

In \cite{DBLP:conf/opodis/AmirAL20} the authors use the \emph{catalytic input model} (CI model) where they have two types of agents: $n$ catalysts and $m$ worker agents ($N = n+m$).
They solve the approximate majority problem for two opinions \whp in $\LDAUOmicron{N\log N}$ interactions in the CI model when the initial bias among the catalysts is $\LDAUOmega{\sqrt{N\log N}}$ and $m = \ldauTheta{n}$.
They show that the size of the initial bias is tight up to a $\LDAUOmicron{\sqrt{\log N}}$ factor.
Additionally, they consider the approximate majority problem in the CI model and in the population model with leaks.
Their protocols tolerate a leak rate of at most $\beta = \LDAUOmicron{\sqrt{N\log N} / N}$ in the CI model and a leak rate of at most $\beta = \LDAUOmicron{\sqrt{n\log n} / n}$ in the population model.
They also show a separation between the computational power of the CI model and the population model.

\section{Population Model and Problem Definitions}
In the population model we are given a set $V$ of $n$ anonymous \emph{agents}.
At each time step two agents are chosen independently and uniformly at random randomly to \emph{interact}.
We assume that interactions between two agents $(u, v)$ are ordered and call $u$ the \emph{initiator} and $v$ the \emph{responder}.
The interacting agents update their states according to a common transition function of their previous states.
Formally, a \emph{population protocol} is defined as a tuple of a finite set of \emph{states} $Q$, a \emph{transition function} $\delta: Q \times Q \rightarrow Q \times Q$, a finite set of \emph{output symbols} $\Sigma$, and an \emph{output function} $\omega: Q \rightarrow \Sigma$ which maps every state to an output.
A \emph{configuration} is a mapping $C: V \rightarrow Q$ which specifies the state of each agent.
An execution of a protocol is an infinite sequence $C_0,C_1,\ldots$ such that for all $C_i$ there exist two agents $v_1, v_2$ and a transition $(q_1,q_2) \rightarrow (q'_1,q'_2)$ such that $C_i(v_1) = q_1, C_i(v_2) = q_2, C_{i+1}(v_1) = q'_1,  C_{i+1}(v_2) = q'_2$ and $C_i(w) = C_{i+1}(w)$ for all $w \neq v_1,v_2$.
The main quality criteria of a population protocol are the required number of states and the running time.
The required number of states is given by the size of the state space $Q$, and the running time is given  by the number of interactions.

\paragraph{Phase Clocks}
Phase clocks are used to synchronize population protocols.
We assume a phase clock is implemented by simple counters $\Clock{u_1},\dots,\Clock{u_n}$ modulo $|Q|$ (see, e.g., \cite{DBLP:conf/soda/AlistarhAG18,DBLP:journals/dc/AngluinAE08a,DBLP:journals/dc/BerenbrinkEFKKR21,10.1145/1017460.1017463,DBLP:conf/soda/GasieniecS18}).
Whenever $\Clock u$ crosses zero, agent $u$ receives a so-called \emph{signal}.
These signals divide the time into \emph{phases} of $\Theta(n \log n)$ interactions each.
We say that a $(\tau, w)$-phase clock is synchronous in the time interval $[t_1, t_2]$ if every agent gets a signal every $\Theta(n\log n)$ interactions.
More formally:
\begin{itemize}[nosep]
\item Every agent receives a signal in the first $2 \cdot (w+1) \cdot \tau \cdot n$ steps of the interval. 
\item Assume an agent $u$ receives a signal at time $t\in [t_1, t_2]$.
\begin{itemize}[nosep]
\item For all $v\in V$, agent $v$ receives a signal at time $t_v$ with $|t - t_v| \leq \tau \cdot n$.
\item Agent $u$ receives the next signal at time $t'$ with $ (w+1) \cdot \tau \cdot n \leq |t - t'| \leq 2 \cdot (w+1) \cdot \tau \cdot n$.
\end{itemize}
\end{itemize}
The above definition  divides the time interval  $[t_1, t_2]$  into a sequence of subintervals that alternates between so-called burst-intervals and overlap-intervals.
\begin{itemize}[nosep]
\item A burst-interval has length at most $\tau \cdot n$ and every agent gets exactly one signal.
\item An overlap-interval consists of those time steps between two burst-intervals where none of the agents gets a signal.
It has length at least $w \cdot \tau \cdot n$.
\end{itemize}
A burst-interval together with the subsequent overlap-interval forms a \emph{phase}.

%\medskip

The goal of this paper is to develop a phase clock that is \emph{loosely-stabilizing} according to the definitions of \cite{DBLP:journals/tcs/SudoNYOKM12}.
To formally define loosely-stabilizing phase clocks, we first define the set of \emph{synchronous} configurations $\mathcal{C}$.
Intuitively, we call a state $C_t$ of a $(\tau, w)$-phase clock at time $t$ synchronous if the counters of all pairs of agents do not deviate much.
More precisely, $\Clock u(t) - \Clock v(t) \leq_{|Q|} (7 + 2 \cdot \sqrt{ 10 + w }) \cdot \tau$ for all pairs of agents $(u,v)$
(Here, ``$\leq_{|Q|}$'' denotes smaller w.r.t.\ the circular order modulo $|Q|$.)
We give the formal definition of a synchronous configuration in the next section.

We now define loosely-stabilizing phase clocks as follows.
Consider an infinite sequence of configurations $C_0,C_1,\ldots $.
For an arbitrary configuration $C_i \not\in \mathcal C$ the \emph{convergence time} is defined as the smallest $t$ such that $C_{i+t_1} \in \mathcal{C}$.
Intuitively, the convergence time bounds the time it takes the clock to reach a synchronous configuration when starting from an asynchronous configuration.
For an arbitrary configuration $C_i\in \mathcal{C}$ the \emph{holding time} $t_2$ is defined as the largest $t$ such that $C_{i+t_2} \in \mathcal{C}$ .
Intuitively, the holding time bounds the time during which the clock remains in a synchronous configuration when starting from a synchronous configuration.
We say that a phase clock is ($t_1, t_2$)-loosely-stabilizing if the maximum convergence time over all possible configurations is \whp less than $t_1$ and the minimum holding time over all synchronous configurations is \whp at least $t_2$.
Note that the probabilities in our bounds are only a function of the randomly selected interaction sequence.

\section{Clock Algorithm}
\label{sec:clock_algorithm}

\noindent\begin{minipage}{\textwidth-150.841pt-2em}
In this section we introduce our phase clock protocol.
Our $(\tau, w)$-phase clock has a state space $Q=\set{0,\dots,\left( 21+w+6 \cdot \sqrt{10+w} \right) \cdot \tau-1}$.
The clock states are divided into 
$\left( 21+w+6 \cdot \sqrt{10+w} \right)$ \emph{hours}, and each hour consists of $\tau = 36 \cdot (c+4) \cdot \ln n$ \emph{minutes}.
The parameter $c \geq 0$ is a constant that is defined in \cref{thm:synchronous_phase_clock}.
As we will see, $\tau$ is a multiple of the running time of the \emph{one-way epidemic} (see \cref{lem:one-way-epidemic}) and
$w \cdot \tau \cdot  n$ is the number of interactions in which our clocks are synchronized.
\end{minipage}\hfill\begin{minipage}{150.841pt}
\begin{figure}[H] 
\vspace{-10ex}
\includegraphics{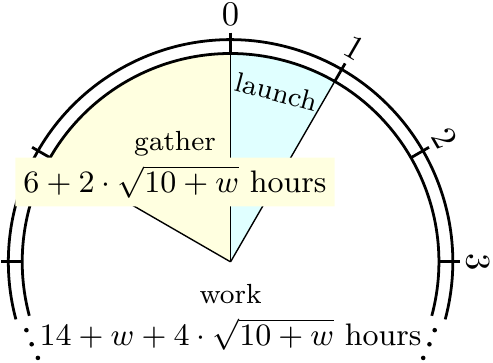}
\caption{\label{fig:clock} Schematic representation of the clock states.}
\end{figure}
\end{minipage}

\smallskip

We divide the hours into three consecutive intervals (see \cref{fig:clock}): the launching interval $I_{launch}$ (first hour), the working interval $I_{work}$ ($14 + w + 4 \cdot \sqrt{10+w}$ hours) and the gathering interval $I_{gather}$ (last $6 + 2 \cdot \sqrt{10+w}$ hours). 
We say that agent $u$ is in one of the intervals whenever its clock counter $\Clock u$ is in that interval.
If the agents are either all in $I_{gather}$, all in $I_{work}$, or all in $I_{launch}$, we say the configuration is \emph{homogeneous}.
Finally, for two agents $u$ and $v$ we define a \emph{distance} $d(u,v) = \min \{ |\Clock{u} -\Clock{v}|, |Q| - |\Clock{u} -\Clock{v}| \}$ % and a relation $u \leq_{|Q|} v \equiv (\Clock{u} \leq \Clock{v} ~ \text{\textsc{xor}} ~ \abs{\Clock{u} - \Clock{v}} > |Q|/2)$ 
that takes the cyclic nature of the clock into account.
This allows us to formally define synchronous configurations as follows.
%\footnote{$|A|$ denotes the cardinality of set $A$.}
%smaller w.r.t.\ the circular order modulo $|Q|$ is defined as
%$a \leq_{(|Q|)} b \equiv (a \leq b~ \text{\textsc{xor}} ~ \abs{a - b} > |Q|/2)$.
\begin{definition*}[Synchronous Configuration] \label{def:synchronous}
A configuration $C$ is called \emph{synchronous} if and only if 
for all pairs of agents $(u,v)$ we have $d(u,v) < (7 + 2 \cdot \sqrt{10+w}) \cdot \tau$.
\end{definition*}
Our clock works as follows.
Assume agents $(u,v)$ interact.
With two exceptions, agent $u$ increments its counter $\Clock u$ by one minute modulo $|Q|$ (\cref{rule:1,rule:2}).
If, however, $u$ is in  $I_{gather}$ and $v$ is in $I_{launch}$ then agent $u$ adopts $\Clock v$ (\cref{rule:3}): we say the agent \emph{hops}.
If $u$ is in $I_{gather}$ and $v$ is in $I_{work}$ then agent $u$ returns to the beginning of the $I_{gather}$ interval  (\cref{rule:4}): we say that the agent \emph{resets} itself.
We define that agent $u$ receives a \emph{signal} whenever its clock crosses the wrap-around from $I_{gather}$ to $I_{launch}$.
Formally, our clock uses the following state transitions.
\begin{align}
\label[rule]{rule:1}
&(q_1,q_2) \in (Q \setminus I_{gather}) \times Q\colon &
&(q_1,q_2) \rightarrow (q_1+1, q_2) & \text{(step forward)}\\
\label[rule]{rule:2}
&(q_1,q_2) \in I_{gather} \times I_{gather}\colon &
&(q_1,q_2) \rightarrow (q_1 + 1 \bmod |Q|, q_2) & \text{(step forward)}\\
\label[rule]{rule:3}
&(q_1,q_2) \in I_{gather} \times I_{launch}\colon &
&(q_1,q_2) \rightarrow (q_2, q_2) & \text{(hopping)}\\
\label[rule]{rule:4}
&(q_1,q_2) \in I_{gather} \times I_{work}\colon &
&(q_1,q_2) \rightarrow (|I_{launch}| + |I_{work}|, q_2) & \text{(reset)}
\end{align}

On an intuitive level, the clock works as follows.
Assume the clock is synchronized and all agents are in $I_{launch}$.
Now consider the next $k = \LdauTheta{n \cdot |Q|}$ interactions.
All agents step forward according to \cref{rule:1} until they reach $I_{gather}$.
The maximum distance between any agents grows during the $k$ interactions but it is still bounded by $\LdauOmicron{\sqrt{k/n}} = \LdauOmicron{\sqrt{|w| \cdot \log n}}$, \whp.
Hence, due to the choice of $w$ there is no agent left behind in $I_{launch}$ when the first agent reaches $I_{gather}$.
Additionally, due to the size of $I_{gather}$ when the last agent enters $I_{gather}$ all of the other agents are still in $|I_{gather}|$.
As soon as the first agent reaches $I_{launch}$, \cref{rule:3} (agents hop onto agents in $I_{launch}$) ensures that all agents start the next phase without a large gap.
Hence, there is an interaction after which all agents are in $|I_{launch}|$ which brings us back to our initial configuration (all agents in $I_{launch}$).

Now we consider an asynchronous configuration $C_t$ where the agents can be arbitrarily distributed over the $|Q|$ states of the clock.
The main idea for the recovery of our clock is as follows.
We show that after $O(n \log n)$ interactions there is a time $t$ where $I_{launch}$ is empty.
After $O(n \log n) $ additional steps most of the agents are in $I_{gather}$:
agents cannot hop since $I_{launch}$ is empty, and they reset as soon as they interact with an agent in $I_{work}$.
They enter $I_{launch}$ as soon as the first agent crosses $0$ by increasing its clock counter.
\medskip

\noindent We will show that the following two properties hold for our clock.

\begin{theorem}
\label{thm:synchronous_phase_clock}
Let $t_1, t_2$ with $t_1 \leq t_2$ be two points in time and assume that the configuration $C_{t_1}$ at time $t_1$ is a homogeneous launching configuration and $t_2-t_1\le n^c$.
Then the clock counters of the agents implement a synchronous $(\tau,w)$-phase clock in the time interval $[t_1, t_2]$ \whp.
\end{theorem}

\begin{theorem}
\label{thm:loosely_stabilizing}
The clock counters of the agents implement a $\left( \LdauOmicron{n \cdot \log n}, \LdauOmega{\text{poly }n} \right)$-loosely-stabilizing $\left(\LdauTheta{\log n},w\right)$-phase clock.
\end{theorem}

\noindent We prove \cref{thm:synchronous_phase_clock} in \cref{sec:maintenance} and \cref{thm:loosely_stabilizing} in \cref{sec:recovery}.

\paragraph{Auxiliary Results}

The one-way-epidemic  is a population protocol with state space $\set{0,1}$ and transitions $(q_1,q_2) \rightarrow (\max \left( q_1, q_2 \right), q_2)$.
An agent in state $0$ is called \emph{susceptible} and an agent in state $1$ is called \emph{infected}.
We say agent $u$ infects agent $v$ if $u$ is infected and initiates an interaction with $v$.
The following result is folklore, see, e.g., \cite{DBLP:journals/dc/AngluinAE08a}.
Additional details can be found in \cref{apx:one-way-epidemic}.

\begin{lemma}[One-way-epidemic]
    \label{lem:one-way-epidemic}
    Assume an agent starts the one-way epidemic in step $1$.
    All agents are infected after $t=\tau/4 \cdot n$ many steps with probability at least $1-n^{-(7+2c)}$.% (see \cref{obs:one-way-epidemic}.
\end{lemma}

The following lemma bounds the number of interactions initiated by some fixed agent $u$ among a sequence of $t$ interactions.
It is used throughout \cref{sec:maintenance, sec:recovery} and follows immediately from Chernoff bounds (see \cref{thm:chernoff-bounds}).
\begin{lemma}
    \label{lem:moving}
    Consider an arbitrary sequence of $t$ interactions and 
    let $X_u$ be the number of interactions initiated by agent $u$ within this sequence.
    Then
    \[
        \Prob{X_u < (1+\delta) \cdot t/n} \geq 1-n^{-\frac{12\cdot(c+4) t\cdot\delta^2}{n\cdot \tau}} \quad \text{and} \quad 
        \Prob{X_u > (1-\delta) \cdot t/n} \geq 1-n^{-\frac{18\cdot(c+4) t\cdot\delta^2}{n\cdot \tau}}.
    \]
\end{lemma}

%\begin{proof}
%    $X_u \sim \BinDistr(x,1/n)$.
%    Since $\Ex{X_u} = x/n$, we can bound the probabilities using \cref{lemma:chernoff_poisson_trials}.
%    Finally, $e^{-x} = e^{-x\cdot \tau/\tau} = e^{-12\cdot(c+4)\ln n \cdot x/\tau} = n^{-12\cdot(c+4)\cdot x/\tau}$.
%\end{proof}

\section{Maintenance: Proof of \cref{thm:synchronous_phase_clock}}
\label{sec:maintenance}

In this section we first show the following main result.
At the end of the section we show how \cref{thm:synchronous_phase_clock} follows from this proposition.

 \begin{proposition}[Maintenance]
     \label{pro:maintenance}
     Consider our $(\tau, w)$-phase clock for $n$ agents with $\tau = 12 \cdot (c+4) \cdot \ln n$ for any $c\geq 6$ and sufficiently large $w$. %$w \geq 114$.
     Let configuration $C_{t_1}$ be a homogeneous launching configuration.
     Then, with probability at least $1-n^{-(c+1)}$, there exists a $t_2 = \LdauTheta{n \cdot w \cdot \log n}$ such that the following holds:
     \begin{enumerate}[nosep]
         \item \label{pro:maintenance:1} $C_{t_1+t_2}$ is a homogeneous launching configuration,
        \item \label{pro:maintenance:2} $\forall t \in \intcc{t_1,t_1+t_2}$: $C_{t}$ is synchronous,
        \item \label{pro:maintenance:3} in the time interval $[t_1, t_1 + t_2]$ 
        there exists a contiguous sequence of homogeneous working configurations of length $w \cdot \tau \cdot n$.
        % in the sequence $C_{t_1},\dots, C_{t_1+t_2}$ of size $w \cdot n$.
         %\item \label{pro:maintenance:4} for each agent $u$ there exists exactly one interaction $t_u \in \intcc{t_1,t_1+t_2}$ where $u$ receives a signal.
     \end{enumerate}
\end{proposition}
We split the proof of \cref{pro:maintenance} into two parts, \cref{lem:launch->gather, lem:gather->launch}.
%\cref{lem:launch->gather} establishes that \whp agents transition from a homogeneous launching configuration into a homogeneous gathering configuration via a sequence of synchronous configurations.
%It also proves the existence of a sufficient overlap and that no agent receives a signal in this sequence.
%\cref{lem:gather->launch} shows that \whp agents transition from a homogeneous gathering configuration into a homogeneous launching configuration via a sequence of synchronous configurations. %where each agent receives exactly one signal.
%The proof is then finalized by a union bound over $n^c$ phases.
The formal proof follows.

\begin{proof}
Assume the configuration $C_{t_1}$ at time $t_1$ is a homogeneous launching configuration.
\Cref{pro:maintenance:1,pro:maintenance:2} of \cref{pro:maintenance} follow immediately from \cref{lem:launch->gather,lem:gather->launch}:
\begin{itemize}[nosep]
\item It follows from \cref{lem:launch->gather} that the agents transition via a sequence of synchronous configu\-rations into a homogeneous gathering configuration within $\Theta(n\cdot w\cdot \log n)$ \mbox{time \whp.}
\item It follow from \cref{lem:gather->launch} that the agents transition via a sequence of synchronous configurations back into a homogeneous launching configuration within $\Theta(n\cdot w\cdot \log n)$ further time \whp.
\end{itemize}
It remains to show \cref{pro:maintenance:3}.
Recall that in a synchronous configuration all pairs of agents have distance (w.r.t.\ the circular order modulo $|Q|$) at most $\Delta = (7 + 2 \cdot \sqrt{10+w}) \cdot \tau$.
Since $|I_{work}| = w \cdot \tau + 2 \Delta $ it immediately follows that there must be $w \cdot \tau \cdot n$ interactions where all agents are in $I_{work}$.
This concludes the proof.
\end{proof}

\noindent The following lemma establishes that \whp all agents transition from a homogeneous launching configuration into a homogeneous gathering configuration via a sequence of synchronous configurations.

\begin{lemma}
    \label{lem:launch->gather}
    Let $C_{t}$ be a homogeneous launching configuration.
    Let $t' = n \cdot \smash{\frac{|I_{launch}|+|I_{work}|}{1-\left(2 \cdot \sqrt{|I_{work}|/\tau} \right)^{-1}}}$.
    %$t' = \LdauTheta{n \cdot w \cdot \log n}$ 
    Then the following holds with probability at least $1 - n^{-(c+1)}/2$:
    \begin{enumerate}[nosep]
        \item \label{lem:launch->gather:1} $C_{t+t'}$ is a homogeneous gathering configuration and
        \item \label{lem:launch->gather:2} $\forall t'' \in \intcc{t,t+t'}: C_{t''}$ is synchronous.
    \end{enumerate} 
\end{lemma}

% \begin{proof}[Proof Sketch]
% We show that the size of $I_{work}$ is chosen sufficiently large such that the following holds.
% Let $t_a$ be the first interaction in which an agent enters $I_{gather}$ and 
% let $t_b$ be the first interaction in which an agent reaches the end of $I_{gather}$.
% We show that \whp $t_a \geq 2 \cdot \tau \cdot n$ and that in $2 \cdot \tau \cdot n$ interactions all agents have already left $I_{launch}$.
% Therefore, agents cannot hop out of $I_{gather}$ before time $t_b$.
% Then, we show that \whp $t_b \geq t'$ and that in $t'$ interactions all agents have reached $I_{gather}$.
% All necessary bounds in this proof follow from application of \cref{lem:moving}.
% The distance between any two agents can also be bounded by application of \cref{lem:moving} and union bounding over $t'$ interactions.
% We show that the reset rule does not increase ...
% \end{proof}

\begin{proof}
In the following we assume w.l.o.g.\  $t = 0$. %...for the sake of brevity of notation that...
We prove the two statements separately.

\paragraph{Statement 1}
Our goal is to show that after $t'$ interactions all agents are in $I_{gather}$ when we start from a homogeneous launching configuration $C_0$ at time $t = 0$.
We first show that there is no agent left in $I_{launch}$ when the first agent enters $I_{gather}$.
Let $t_a$ be the first interaction in which an agent enters $I_{gather}$.
Note that before $t_a$ all agents are either in $I_{launch}$ or in $I_{work}$ and thus the agents increase their counter by one whenever they initiate an interaction.

First we show that \whp $t_a \ge 2 \cdot \tau \cdot n$.    
Let $X_u(2 \cdot\tau \cdot n)$ denote the number of interactions agent $u$ initiates before time $2\cdot\tau \cdot n$.
From \cref{lem:moving} it follows with $\delta=1$ that $X_u(2 \cdot\tau \cdot n) < 4 \cdot \tau$ with probability at least $1-n^{-24\cdot(c+4)}$.
Since $4 \cdot \tau < |I_{work}|$, it holds that $\Clock{u}(2 \cdot\tau \cdot n) < |I_{launch}|+|I_{work}|$ in this case.
Hence, agent $u$ has not yet reached $I_{gather}$ with probability at least $1-n^{-24\cdot(c+4)}$ at time $t_a$.
It follows from a union bound over all agents that no agent has reached $I_{gather}$ with probability at least $1-n^{-24\cdot(c+4)+1}$ at time $t_a$.

Next we show that \whp at time $2 \cdot\tau \cdot n$ all agents have left $I_{launch}$.
As before, let $X_u(2 \cdot\tau \cdot n)$ denote the number of interactions agent $u$ initiates before time $2\cdot\tau \cdot n$.
From \cref{lem:moving} it follows with $\delta=1$ that $X_u(2 \cdot\tau \cdot n) > \tau$ with probability at least $1-n^{-36\cdot(c+4)}$.
Since $\tau = |I_{launch}|$, it holds that $\Clock{u}(t+2 \cdot\tau \cdot n) \ge |I_{launch}|$ in this case.
Hence, agent $u$ has left $I_{launch}$ with probability at least $1-n^{-36\cdot(c+4)}$ at time $t_a$.
Again, it follows from a union bound over all agents that all agents have left $I_{launch}$ with probability at least $1-n^{-36\cdot(c+4)+1}$ at time $t_a$.

Let now $t_b$ be the first interaction in which an agent enters the last minute of $I_{gather}$ and observe that $t_b > t_a$.
Then, \whp no agent is in $I_{launch}$ during the time interval $\intcc{t+t_a,t+t_b}$.
Therefore, agents cannot hop.
Thus, by definition of $t_b$, no agent can leave $I_{gather}$ before time $t_b$.
Agents that initiate an interaction must therefore either increase their counter by one or reset.

First we show that \whp $t_b > t'$.
From \cref{lem:moving} it follows with $\delta = \left(2\cdot\sqrt{|I_{work}|/\tau}\right)^{-1}$  that $X_u(t') < |I_{work}| + |I_{gather}|$ with probability at least $1-n^{-3(c+4)}$.
(Note that we use $(1+\delta)/(1-\delta) < 1/(1-2\cdot\delta)$ for $\delta < 0.5$ and $(|I_{launch}+|I_{work}|)/(1-2\cdot\delta) = |I_{work}| + \frac{w+5\cdot\sqrt{10+w}+16}{\sqrt{10+w}+1}\cdot \tau < |I_{work}| + |I_{gather}|$.)
Thus, $\Clock{u}(t') \leq \Clock{u}(0) + X_u(t') < |I_{launch}| -1 + |I_{work}| + |I_{gather}|$ (which is the last state of $I_{gather}$) with probability at least $1-n^{-3(c+4)}$.
By a union bound, this holds for all agents with probability at least $1-n^{-(3c+11)}$.

Next we show that \whp at time $t_b$ all agents have reached $I_{gather}$.
From \cref{lem:moving} it follows with $\delta = \left(2\cdot\sqrt{|I_{work}|/\tau}\right)^{-1}$ that $X_u(t') > |I_{launch}| + |I_{work}|$ with probability at least $1-n^{-9\cdot(c+4)/2}$.
Thus, $\Clock{u}(t') \geq \Clock{u}(0) + X_u(t') > 0 + |I_{launch}| + |I_{work}|$, with probability at least $1-n^{-9\cdot(c+4)/2}$.
By a union bound, this holds for all agents with probability at least $1-n^{-(17+9/2\cdot c)}$.

Together it follows that at time $t'$ no agent has left $I_{gather}$ but all agents have entered it with probability at least $1-n^{-(c+3)}$.
Therefore, $C_{t'}$ is a homogeneous gathering configuration.

\paragraph{Statement 2}
Recall that a synchronous configuration $C$ is defined as a configuration where $\max_{(u,v)} \set{d(u,v)} < |I_{launch}| + |I_{gather}|$.
As before, let $X_u(i)$ denote the number of interactions agent $u$ initiates before time $i$.
Now fix a time $t \leq t'$ and a pair of agents $(u, v)$ with $X_u(t) < X_v(t)$.
We use \cref{lem:moving} to bound the deviation of $X_u(t)$ and $X_v(t)$ at time $t$ as follows:
$\Prob{X_u(t) > t / n - \left|I_{gather}\right|/2} \geq 1 - n^{-6(c+4)}$ and
$\Prob{X_v(t) < t / n + \left|I_{gather}\right|/2} \geq 1 - n^{-4(c+4)}$.
Therefore, $|X_v(t) - X_u(t)| < |I_{gather}|$ with probability at least $1-n^{-4(c+4)}-n^{-6(c+4)}$.

Note that \cref{lem:moving} allows us to bound the deviation in the numbers of interactions initiated by agents $u$ and $v$.
However, this does not immediately give a bound on the difference of the clock counters $|\Clock v(t) - \Clock u(t)|$.
To bound the deviation of clock counters (by $|I_{launch}| + |I_{gather}|$), we therefore distinguish three cases.

First, assume that neither $u$ nor $v$ have reached $I_{gather}$ at time $t$.
Then $\Clock u(t) = \Clock u(0) + X_u(t)$ and $\Clock v(t) = \Clock v(t) + X_v(t)$.
Observe that by the assumption of the lemma, both $u$ and $v$ are in $I_{launch}$ at time $t = 0$ and thus $|\Clock v(0) - \Clock u(0)| < |I_{launch}|$.
Together with the above bound on $|X_v(t) - X_u(t)|$ we get $|\Clock v(t) - \Clock u(t)| < |I_{launch}| + |I_{gather}|$.

Secondly, assume that $u$ has not reached $I_{gather}$ but $v$ has reached $I_{gather}$ at time $t$.
Then $\Clock u(t) = \Clock u(0) + X_u(t)$.
For $\Clock v(t)$, however, it might have occurred that $v$ has reset in some interactions before time $t$.
Nevertheless, the clock counter of $v$ is bounded by the number of initiated interactions such that
$\Clock v(t) \leq \Clock v(t) + X_v(t)$.
(Note that $v$ can only increment its $\Clock v$ counter or reset its value; hopping is not possible since we have shown in the proof of the first statement that $I_{launch}$ is empty when the first agent enters $I_{gather}$.)
Therefore, we get again $|\Clock v(t) - \Clock u(t)| < |I_{launch}| + |I_{gather}|$.

Finally, assume that both $u$ and $v$ are in $I_{gather}$ at time $t$.
Then $|\Clock v(t) - \Clock u(t)| \leq |I_{gather}| < |I_{launch}| + |I_{gather}|$ is trivially true.

There are no further cases: in the proof of the first statement we have shown that all agents transition from a homogeneous launching configuration to a homogeneous gathering configuration during the time interval $[0, t']$.
The result now follows from a union bound over all $o(n^2)$ points in time $t \leq t'$ and all $n\cdot(n-1)$ pairs of agents.
\end{proof}

\iffalse
\begin{observation}
    \label{obs:launch->gather}
    Let $C_{t}$ be a homogeneous launching configuration.
    Let $t' = n \cdot \frac{|I_{launch}|+|I_{work}|}{1-\left(2 \cdot \sqrt{|I_{work}|/\tau} \right)^{-1}}$.
    Then the following holds with probability at least $1 - n^{-(c+1)}/2$:
    \begin{enumerate}
        \item \label{lem:launch->gather:3} there exists a sub-sequence of consecutive homogeneous working configurations in the sequence $C_{t},\dots, C_{t+t'}$ of size $w \cdot n$ and
        \item \label{lem:launch->gather:4} no agent receives a signal in $\intcc{t,t+t'}$.
    \end{enumerate}
\end{observation}

\begin{proof}
Note that $t' \cdot n$ < $2\cdot n \cdot |I_{work}|$. 
If \cref{lem:launch->gather} holds, then all agents have crossed $I_{work}$ exactly once and no agent has crossed the wrap-around from $I_{gather}$ to $I_{launch}$.
\end{proof}
\fi

The following lemma is the main technical contribution of this section.
It establishes that \whp all agents transition from a homogeneous gathering configuration into a homogeneous launching configuration via a sequence of synchronous configurations.
Consider a homogeneous gathering configuration and recall that whenever an agent hops from $I_{gather}$ into $I_{launch}$ it adopts the state of the responder.
The main difficulty is to show that all agents hop into $I_{launch}$ before the first agent leaves $I_{launch}$.

\begin{lemma}
    \label{lem:gather->launch}
    Let $C_{t}$ be a homogeneous gathering configuration.
    Then with probability at least $1 - n^{-(c+1)}/2$ the following holds:
    \begin{enumerate}
        \item \label{lem:gather->launch:1} there exists a $t_0 = \LdauOmicron{n \cdot \sqrt{w} \cdot \log n}$ such that the first agent enters $I_{launch}$ at time $t + t_0$,
        \item \label{lem:gather->launch:2} there exists a $t' \leq \tau/4 \cdot n$ such that $C_{t + t_0+t'}$ is a homogeneous launching configuration,
        \item \label{lem:gather->launch:3} $\forall t'' \in \intcc{t,t+t'}: C_{t''}$ is synchronous.
    \end{enumerate} 
\end{lemma}

\begin{proof}
First we prove that \whp there exists a homogeneous launching configuration $C_{t'}$.

\paragraph{Statement 1}
Let $t_0$ be defined such that the first agent $u$ leaves $I_{gather}$ at time $t + t_0$.
Since $C_t$ is a homogeneous gathering configuration, $I_{launch}$ is empty at time $t$ and hence agent $u$ can only leave $I_{gather}$ by increasing its counter.
In every interaction before time $t+t_0$ some agent has to increase its state by one.
Thus $t_0 \leq n \cdot |I_{gather}| = \LdauOmicron{n \cdot \sqrt{w} \cdot \log n}$.

\paragraph{Statement 2}
We continue our analysis at time $t_0$ and again assume w.l.o.g.\ for the sake of brevity of notation that $t_0=0$.
Note that at that time exactly one agent is in state $0$ and all remaining agents are still in $I_{gather}$.
We show the following: there exists a time $\tilde{t}=\tau/4 \cdot n$ such that at time $\tilde{t}$ all agents are in $I_{launch}$ (Recall that $|I_{launch}| = \tau \cdot n$).%the last agent $I_{launch}$ and all of the other agents are still in $I_{launch}$.
To do so we first define a simplified process with the same state space $Q$, however, we refer to the last state of $I_{launch}$ as $\Stop$.
Agents in $\Stop$ never change their state (which renders the states of $I_{work}$ unreachable).
The formal definition of the simplified process is as follows.
\begin{align*}
    %\label[rule]{rule:1mod}
    & (q_1,q_2) \in (I_{launch} \setminus \set{\Stop}) \times Q\colon &
    & (q_1,q_2) \rightarrow (q_1+1, q_2) & \text{(step forward)}\\        
    %\label[rule]{rule:2mod}
    & (q_1,q_2) \in I_{gather} \times I_{gather}\colon &
    & (q_1,q_2) \rightarrow (q_1+1, q_2) \bmod |Q| & \text{(step forward)}\\
    %\label[rule]{rule:3mod}
    & (q_1,q_2) \in I_{gather} \times I_{launch}\colon &
    & (q_1,q_2) \rightarrow (q_2, q_2) & \text{(hopping)}\\
    %\label[rule]{rule:4mod}
    & (q_1,q_2) \in \set{\Stop} \times Q\colon &
    & (q_1,q_2) \rightarrow (q_1, q_2) & \text{(stopping)}
\end{align*}
For this simplified process we show a lower bound: after $\tilde{t} = \LdauTheta{n \cdot \log n}$ interactions all agents are in $I_{launch}$.
Then we show (for the simplified process) an upper bound: in $C_{\tilde{t}}$ none of the agents are in state $\Stop$.
A simple coupling of the simplified process and the original process shows that under these circumstances none of the agents entered $I_{work}$ for our original process.
This finishes the proof with $t'=\tilde{t}$.

\medskip\noindent
\textit{Lower Bound.}\quad
In the simplified process agents can enter $I_{launch}$ either via hopping or by making enough steps forward on their own.
From \cref{lem:one-way-epidemic} it follows that all agents enter $I_{launch}$ after at most $\tilde{t}=\tau/4 \cdot n$ interactions with probability at least $1-n^{-(5+c)}$.
(For the upper bound, one can simply discard setting the clock counter to zero when an agent enters $I_{launch}$ by increasing its counter.)
Showing that none of the agents are in state $\Stop$ is much harder.
Due to the hopping the clock counters of agents in $I_{launch}$ are highly correlated.
Nevertheless, we can show that the clock counters of each agent can be majorized by independent binomially distributed random variables as follows.

\medskip\noindent
\textit{Upper Bound.}\quad
Let $u_i$ be the $i$'th agent that enters $I_{launch}$ and let $t_i$ be the time when $u_i$ enters $I_{launch}$.
Let furthermore $X_i(t)$ be a random variable for the clock counter of agent $u_i$ in $I_{launch}$ in the time interval $[0, t]$.
Formally, we define for a time step $t$ that $X_i(t) = 0$ if $u_i$ is in $I_{gather}$ and $X_i(t) = \Clock{u_i}(t)$ if $u_i$ is in $I_{launch}$.
We show by induction on $i$ that $X_i(t)$ is majorized by a random variable $Z_i(t)$ with binomial distribution $Z_i(t) \sim \BinDistr\left(t, 1/n\cdot(1 + 1/(n-1))^{i-1}\right)$.
Ultimately, our goal is to apply Chernoff bounds to $Z_i(\tilde{t})$ which shows that agent $u_i$ does not reach $\Stop$ \whp.
The statement for the simplified process then follows from a union bound over all agents \whp.

\medskip\noindent
\textit{Base Case.}\quad
For the base case we consider all agents that enter $I_{launch}$ on their own by incrementing their counters to $0$ (modulo $|Q|$) in $I_{gather}$.
Fix such an agent $u_i$.
It holds that $X_i(t)$ for $t \geq t_i$ has binomial distribution $X_i(t) \sim \BinDistr(t - t_i, 1/n)$.
Therefore, $X_i(t) \prec Z_i(t)$ as claimed.\footnote{The expression $X \prec Y$ means that the random variable $X$ is majorized by the random variable $Y$.}
(Intuitively, this means that the clock counter of any other agent $u_i$ with $i > 1$ that enters $I_{launch}$ at time $t_i > 0$ is majorized by the clock counter of an agent which enters $I_{launch}$ at time $t_1 = 0$ and increments its counter with probability $1/n$.)

\medskip\noindent
\textit{Induction Step.}\quad
For the induction step we now consider all agents that enter $I_{launch}$ by hopping onto some other agent in $I_{launch}$.
Fix such an agent $u_i$.
Let $S_{i}$ be the event that agent $u_i$ is the $i$'th agent that enters $I_{launch}$.
Let furthermore $t_i$ be the time when $u_i$ enters $I_{launch}$.
We condition on $S_{i}$ and observe that agent $u_i$ enters $I_{launch}$ by hopping onto some other agent $u_j \in \set{u_1, \dots, u_{i-1}}$.
Intuitively, we would now like to exploit the fact that the counter of agent $u_i$ is copied at time $t_i$ from agent $u_j$ such that $X_i(t_i) = X_j(t_i)$.
Unfortunately, we must be extremely careful here: conditioning on $S_{i}$ alters the probability space!
(For example, under $S_{i}$ the agent $u_i$ with $i \geq 3$ cannot initiate an interaction with agent $u_1$ before agent $u_2$ does, since $S_{i}$ rules out that $u_i$ enters $I_{launch}$ before agent $u_2$.)
We account for the modified probability space as follows.

Let $\Omega_{S_{i}}(t)$ be the probability space of possible interactions conditioned on $S_{i}$ at time $t \leq \tilde{t}$.
Without the conditioning on $S_i$, the probability space $\Omega(t)$ at time $t$ contains all (ordered) pairs of agents with $|\Omega(t)| = n \cdot (n-1)$.
When conditioning on $S_i$, the event $S_i$ rules out that agent $u_i$ interacts with any other agent $u_j \in I_{launch}$ before time $t_i$.
In particular, agent $u_i$ cannot interact with another agent $u_j$ with $j < i$ during the time interval $[t_j, t_i]$.
In order to give a lower bound on $|\Omega_{S_{i}}(t)|$, we exclude all $(n-1)$ interactions $(u_i, u_j)$ for $j \in [n]$ from $\Omega(t)$.
Hence $|\Omega_{S_{i}}(t)| \geq n \cdot(n-1) - (n - 1) = (n-1)^2$ for any time $t \leq t_i$.
(The probability space after time $t_i$ is not affected by conditioning on $S_i$, but the majorization holds nonetheless.)
We now consider the event $E_{\hat{t}}$ for $\hat{t} \leq t_i$ that the interaction at time $\hat{t}$ increments $X_j(t)$ by $1$ (recall that $u_j$ is the agent onto which $u_i$ hopped).
It then holds for the reduced probability space $\Omega_{S_{i}}$ that $\Prob{E_{\hat{t}} | S_{i}} \leq \Prob{E_{\hat{t}}} \cdot |\Omega(t)| / |\Omega_{S_{i}}(t)|$.
(Note that $\Omega_{S_{i}}$ is still a uniform probability space.)
We calculate
\[
\frac{ |\Omega(t)|}{|\Omega_{S_{i}}(t)|} = \frac{n\cdot(n-1)}{(n-1)^2} = 1 + \frac{1}{n-1}
\]
and get $\Prob{E_{\hat{t}} | S_i} \leq \Prob{E_{\hat{t}}} \cdot (1 + {1}/{(n-1)})$ for $\hat{t} \leq t_i$.
Therefore, we use the induction hypothesis (that describes $X_{j}(t_i)$) and get $X_{i}(t_i) \prec Z_{i}(t_i)$, where $Z_{i}(t_i) \sim \BinDistr\left(t_i, 1/n \cdot \left(1 + 1/(n-1)\right)^{i-1}\right)$.
Similarly, we define $E_{\hat{t}}$ for $\hat{t} \geq t_i$ to be the event that $u_i$ increments its counter in $I_{launch}$.
Observe that $\Prob{E_{\hat{t}}} \leq 1/n$ for $\hat{t} > t_i$.
It follows that $X_{i}(t) \prec Z_{i}(t)$ with distribution $Z_{i}(t) \sim \BinDistr(\tilde{t}, 1/n \cdot (1+1/(n-1))^{i-1})$ for $t \leq \tilde{t}$ as claimed.
This concludes the induction.

\medskip\noindent
\textit{Conclusions.}\quad
From the induction it follows that for each agent $u_i$ the clock counter $\Clock{u_i}(\tilde{t})$ at time $\tilde{t}$ is majorized by a random variable $Z(\tilde{t})$ with binomial distribution $Z(\tilde{t}) \sim \BinDistr(\tilde{t}, e/n)$.
(Note that we used the inequality $(1+1/(n-1))^{(n-1)} < e$.)
From \cref{thm:chernoff-bounds} it follows that $\Prob{Z(\tilde{t}) \geq \tau-1 } \leq n^{-(c+4)}$.
Finally, the proof for the simplified process follows from a union bound over all agents.

\medskip

It is now straightforward to couple the actual phase clock process with the simplified process.
Assume that we start both processes at time $0$ when exactly one agent is in state $0$.
In the simplified process no agent reaches state $\tau$ in $\tau/4\cdot n$ interactions with probability at least $1-n^{-(c+3)}$.
In this case, however, the simplified process and the actual phase clock process do not deviate and, in particular, no agent reaches the beginning of $I_{work}$ in $\tau/4\cdot n$ many interactions.
Thus, the configuration $C_{t'}$ is a homogeneous launching configuration with probability at least $1-n^{-(c+3)}$.

\paragraph{Statement 3}
By definition, all configurations where all agents are in $I_{gather} \cup I_{launch}$ are synchronous.
\end{proof}

\noindent We are now ready to put everything together and prove our first theorem.

\begin{proof}[Proof of \cref{thm:synchronous_phase_clock}]
The proof of \cref{thm:synchronous_phase_clock} follows readily from the main result of this section, \cref{pro:maintenance}.

Assume the configuration at time $t_1$ is a synchronous launching configuration.
Then from \cref{pro:maintenance} it follows \whp that after $t_2 = \Theta(n \cdot w \cdot \log n)$ interactions the configuration $C_{t_2}$ is again a homogeneous launching configuration, and all configurations in $[t_1, t_2]$ are synchronous.
From \cref{pro:maintenance:3} it follows that no agent receives a signal in a contiguous subinterval $[t_1', t_2'] \subset [t_1, t_2]$ of length $t_2' - t_1' = w \cdot \tau \cdot n$.
This shows that we have \whp the required \emph{overlap} according to the definition of synchronous $(\tau,w)$-phase clocks.

From \cref{lem:gather->launch} it follows \whp that all agents transition from a homogeneous gathering configuration into a homogeneous launching configuration within $\tau/4 \cdot n $ interactions.
Recall that whenever an agent crosses zero, it receives a signal.
Therefore, when all agents transition from a homogeneous gathering configuration into a homogeneous launching configuration via a sequence of synchronous configurations, all agents receive exactly one signal, and the time between two signals of two agents $(u, v)$ is \whp at most $\tau/4 \cdot n $.
This shows that we have \whp the required \emph{bursts} according to the definition of synchronous $(\tau,w)$-phase clocks.

Together, the counters of our clock implement a synchronous $(\tau,w)$-phase clock in $[t_1, t_2]$ with probability $n^{-c}$.
It follows from an inductive argument that the clock counters implement a synchronous $(\tau,w)$-phase clock during the $n^c$ interactions that follow time $t_1$ \whp.
\end{proof}

\section{Recovery: Proof of \cref{thm:loosely_stabilizing}}
\label{sec:recovery}

In this section we first show the following main result.
At the end of the section we show how \cref{thm:loosely_stabilizing} follows from this proposition.

\begin{proposition}[Recovery]
    \label{thm:recovery}
    Consider our $(\tau, w)$-phase clock with $n$ agents and sufficiently large $c$ and $w$. %$c\geq 6$, $w \geq 114$.
    Let $C_{t_1}$ be an arbitrary configuration.
    Then with probability at least $1-1/n$, there exists a $t_2 = \LdauOmicron{n \cdot w \cdot \log n}$ such that $C_{t_1+t_2}$ is a homogeneous launching configuration.
\end{proposition}

We say a configuration is an \emph{almost homogeneous gathering configuration} if no agent is in $I_{launch}$ and at least $0.9\cdot n$ many agents are in $I_{gather}$.
We start our analysis by showing that within $t=\LdauOmicron{n \cdot w \cdot \log n}$ interactions, we reach an almost homogeneous gathering configuration $C_{t_1+t}$.

\begin{lemma}[name=,restate=lemmaRecoveryOne,label=lem:any->emptylaunch]
    Let $C_{t}$ be an arbitrary configuration.
    Then with probability at least $1-1/(3n)$, there exists a $t' = \LdauOmicron{n \cdot w \cdot \log n}$ such that $C_{t+t'}$ is an almost homogeneous gathering configuration.
\end{lemma}
\begin{proof}[Proof Sketch]
    The main idea of the proof is as follows.
    If there are not too many agents in $I_{gather}$, the reset rule prevents agents from reaching the end of $I_{gather}$.
    Agents may still enter $I_{launch}$ by hopping, but if no agent enters state $0$, eventually there is no agent left in state $0$ to hop on.
    Then the same argument applies to state $1$, and so on.
    Eventually, there are no agents left in $I_{launch}$ to hop onto.
    This means the agents are \emph{trapped} in $I_{gather}$ until a sufficiently large number of agents enters $I_{gather}$ which renders resetting quite unlikely again.
    The resulting configuration is what we call an almost homogeneous gathering configuration.
\end{proof}

Next, we show that from an almost homogeneous gathering configuration we reach a homogeneous gathering configuration in $\LdauOmicron{n \cdot w \cdot \log n}$ interactions. From \cref{lem:gather->launch} in \cref{sec:maintenance} it then follows that we reach a homogeneous launching configuration in an additional number of $O(n \cdot \log n)$ interactions.

\begin{lemma}[name=,restate=lemmaRecoveryTwo,label=lem:recovery:gather->launch]
    Let $C_{t}$ be an almost homogeneous gathering configuration.
    Then with probability at least $1-1/(3n)$, there exists a $t' = \LdauTheta{n \cdot w \cdot \ln n}$ such that $C_{t+t'}$ is a homogeneous gathering configuration.
\end{lemma}

\begin{proof}[Proof Sketch]

If $C_{t}$ is an almost homogeneous gathering configuration, then there are no agents in $I_{launch}$ and at least $0.9\cdot n$ many agents in $I_{gather}$.
Thus, agents cannot hop until an agent enters $I_{launch}$ on its own.
Now there are two cases.
If no agent enters $I_{launch}$ on its own before the last agent enters $I_{gather}$, we are done: this is by definition of a homogeneous gathering configuration.
Otherwise, we will show that a large fraction of agents leave $I_{gather}$ together.
This large fraction behaves similar as in the proof of the maintenance.
The remaining agents have a small \emph{head start} but then they are again \emph{trapped} in $I_{gather}$ until the bulk of agents arrives.
Once the bulk of agents enters $I_{gather}$ we have reached a homogeneous gathering configuration and all agents start to run through the clock synchronously.
\end{proof}

\noindent We are now ready to put everything together and show our second main theorem.

\begin{proof}[Proof of \cref{thm:loosely_stabilizing}]
The proof of \cref{thm:loosely_stabilizing} follows readily from the main result of this section, \cref{thm:recovery}.
Observe that $\tau = \Theta(\log n)$.
According to \cref{thm:recovery}, our clock recovers to a homogeneous launching configuration in $O(n\cdot \log n)$ interactions.
By \cref{thm:synchronous_phase_clock}, this marks the beginning of a time interval in which the agents implement a synchronous $(\tau,w)$-phase clock.
It follows immediately from \cref{thm:synchronous_phase_clock} that this interval has length $n^c$.
Together, this implies that our $(\tau, w)$-phase clock is a $(O(n\cdot log n), \Omega(\poly n))$-loosely-stabilizing $(\Theta(\log n), w)$-phase clock.
\end{proof}

\section{Adaptive Majority Problem}
\label{sec:majority}

\let\Work\Opinion
\def\IN{\text{\textsc{in}}}

In this section we consider the \emph{adaptive majority problem}
introduced in \cite{DBLP:conf/podc/Alistarh0U21} under the name \emph{robust comparison} and defined as follows.
At any time, every agent has as input either an \emph{opinion} ($A$ or $B$) or it has no input, in which case we say it is \emph{undecided} ($U$).
During the execution of the protocol, the inputs to agents can change.
In the adaptive majority problem, the goal is that all agents output the opinion which is dominant among all inputs.
In this setting we present a \emph{loosely-stabilizing} protocol that solves the adaptive majority problem.

Recall that the performance of a loosely-stabilizing protocol is measured in terms of the \emph{convergence time} and the \emph{holding time}.
Note that the loose-stabilization comes from an application of our phase clock (see \cref{sec:clock_algorithm}).
The phase clocks guarantee synchronized phases for polynomial time. 
During this time we say a configuration $C$ is \emph{correct} w.r.t.\ the adaptive majority problem if the following conditions hold.
Suppose there is a sufficiently large bias towards one opinion. Then
every agent in a correct configuration outputs the majority opinion.
Otherwise, if there is no sufficiently large bias, we consider any output of the agents as correct.
% Similar to loosely-stabilizing phase clocks we can define the convergence time and holding time  for the adaptive majority problem over correct configurations. 
In this setting, we show the following result:
We show that a $(\LdauOmicron{\log n}, \poly n)$-loosely-stabilizing algorithm exists that solves adaptive majority, using $\LdauOmicron{\log n}$ states per agent.

\subsection{Our Protocol}

Our protocol is based on the $(\tau, w)$-phase clock defined in \cref{sec:clock_algorithm} with $w = 566$. % and sufficiently large $c$. (Recall that $\tau = \tau(c)$).
In addition to the states required by the clock, every agent $v$ has three variables $\Input v$, $\Opinion v$, and output $\Output v$.
The variable $\Input v$ always reflects the current input to the agent, $\Opinion v$ holds the current opinion of agent $v$, and $\Output v$ defines the current output value of agent $v$.
All three variables take on values in $\set{A, B, U}$, where $A$ and $B$ stand for the corresponding opinions and $U$ stands for \emph{undecided}.

We use the $(\tau,w)$-phase clock to synchronize the agents.
Then it follows from \cref{pro:maintenance} that all configurations are synchronous \whp.
Observe that in a synchronous configuration for our choice of parameters the clock counters of agents do not deviate by more than $\Delta = 55 \cdot \tau$.
This allows us to define three \emph{subphases} of $I_{work}$, where agents execute three different protocols, as follows.
We split the working interval $I_{work}$ into six contiguous subintervals of equal length.
The clock counters $\Clock u$ allows us to define a simple interface to the phase clock for each agent $u$ as follows.
The variable $\Subphase u$ for each agent $u$ is then defined as follows.
We set $\Subphase u = 1$ if $\Clock u$ is in the first subinterval of $I_{work}$,
$\Subphase u = 2$ if $\Clock u$ is in the third subinterval of $I_{work}$, and
$\Subphase u = 3$ if $\Clock u$ is in the fifth subinterval of $I_{work}$.
Otherwise, $\Subphase u = \bot$.
The clock now assures a clean separation into these subphases such that no two agents perform a different protocol at any time \whp.
Additionally, we will show the overlap within each subphase is long enough such that the subprotocols for the corresponding subphases succeed \whp.

On an intuitive level, our protocol works as follows.
At the beginning of the phase, the input is copied to the opinion variable.
In the first protocol, the support of opinions $A$ and $B$ is amplified until no undecided agents are left.
We call this the Pólya Subphase.
In the second protocol, agents with opposite opinions cancel each other out, becoming undecided.
We call this the Cancellation Subphase.
Finally, in the third protocol the single remaining opinion is amplified again.
We call this the Broadcasting Subphase.
The resulting opinion is copied to the output variable after the working interval $I_{work}$.
Formally, our protocol is specified in \cref{alg:adaptive-majority}.

\begin{algorithm}[ht]

update $\Clock u$ according to \cref{rule:1,rule:2,rule:3,rule:4} with $w=566$\;
\If{Agent $u$ receives a signal}{
$\Work u \gets \Input u$\;
}
\If{$\Subphase  u = 1 \land \Work u = U$}{
$\Work u \gets \Work v$\;
}
\If{$\Subphase  u = 2 \land \Work u \neq \Work v \land \Work u, \Work v \neq U$}{
$\Work u, \Work v \gets U$\;
}
\If{$\Subphase u = 3 \land \Work u = U$}{
$\Work u \gets \Work v$\;
}
\If{$\Clock u \geq |I_{launch}| + |I_{work}|$}{
$\Output u \gets \Work u$\;
}
\caption{Interaction of agents $(u,v)$ in the adaptive majority protocol.}
\label{alg:adaptive-majority}
\end{algorithm}

\paragraph{Result and Notation}
In the remainder of this section, we let $A_t$ and $B_t$ denote the number of agents $u$ with $\Opinion u = A$ and $\Opinion u = B$, respectively, at time $t$. Analogously, we let $A_t^{\IN}$ and $B_t^{\IN}$ denote the number of agents $u$ with $\Input u = A$ and $\Input u = B$, respectively, at time $t$.
We now state our main result for this section, where we assume w.l.o.g.\ that $A$ is the majority and $B$ is the minority opinion.

\begin{theorem}
\label{thm:adaptive-majority}
% Assume that at time $t_0$ we have $A_{t_0}^{\IN} \geq c\cdot\log n$ and $A_{t_0} \geq \left(1+\beta\right)\cdot B_{t_0}$. If $c$ and $\beta$ are large enough constants, 
% then \cref{alg:adaptive-majority} is $(t_1,t_2)$-loosely-stabilizing such that all agents output $A$ during the time interval $[t_1, t_2]$ \whp, where the convergence time $t_1 = \LdauOmicron{n\log n}$ and the holding time $t_2 = \poly n$ \whp. \ftd{'...all agents output A during the time interval $[t_0+t_1, t_0+t_1+t_2]...'$}

\cref{alg:adaptive-majority} is a $(\LdauOmicron{n\log n}, \LDAUOmega{\poly{n}})$-loosely stabilizing adaptive majority protocol. 
\end{theorem}

\subsection{Analysis}

In the following analysis, we consider an arbitrary but fixed phase. 
We condition on the event that the clock is synchronized according to \cref{pro:maintenance} during that phase.
We show the following main result, and later in this section we describe how \cref{thm:adaptive-majority} follows from it this proposition. 
The proofs for the statements in this section can be found in \cref{apx:majority}.

\begin{proposition}[name=,restate=promajority,label=pro:adaptive-majority]
Assume that at time $t_1$ the clocks are in a homogeneous launching configuration and we have $A_{t_1} \geq \alpha \cdot\log n$ and $A_{t_1} \geq \left(1+\beta\right)\cdot B_{t_1}$.
If $\alpha$ and $\beta$ are large enough constants, then
there exists a $t_2 = \Theta(n \cdot w \cdot \log n)$
such that all agents output $A$ in configuration $C_{t_1 + t_2}$
with probability $1 - n^{-c}$.
\end{proposition}

The analysis is split into three parts, one for the \emph{Pólya Subphase}, one for the \emph{Cancellation Subphase}, and one for the \emph{Broadcasting Subphase}.
First, we assume that no changes in the input occur. This highlights the simplicity of the application of our phase clock.
Then we generalize our results: we adopt the undecided state dynamics introduced in \cite{DBLP:journals/dc/AngluinAE08}, and show how we can tolerate input changes at various rates.%a rate of $O(1/n)$.

Observe that from the guarantees of the phase clock in \cref{thm:synchronous_phase_clock} we get a strict separation between the subphases:
no two agents can be more than $1/6$ of $I_{work}$ apart. %in different subphases at the same time.
Furthermore we know that every agent has copied its input at the beginning of the phase before the first agent enters the first subphase.
At last the total time for the three subphases (including the separation time) is sufficiently large such that every agent has finished its work before the next phase starts.

When we refer to a distribution \emph{before a subphase}, we mean the distribution at the time just before the first agent performs an interaction in that subphase.
Analogously, when we refer to a distribution \emph{after a subphase}, we mean the distribution at the time when the last agent has performed an interaction in that subphase.
Recall that in the following analysis, we let $A_t$ and $B_t$ denote the number of agents $u$ with $\Opinion u = A$ and $\Opinion u = B$, respectively, at time $t$.
Furthermore, we let $s_i$ and $e_i$ (for \emph{start} and \emph{end}) be the first and the last time, respectively, when an agent performs an interaction in the $i$'th subphase.

\paragraph{Subphases}
We first consider the Pólya Subphase, where we model the process by means of so-called Pólya urns.
Pólya urns are defined as follows.
Initially, the urn contains $a$ red balls and $b$ blue balls.
In each step, a ball is drawn uniformly at random from the urn.
The ball's color is observed, and it is returned into the urn along with an additional ball of the same color.
The Pólya-Eggenberger distribution $\PE(a,b,m)$ describes the total number of red balls after $m$ steps of this urn process.

This observation allows us to apply concentration bounds to the opinion distribution after the Pólya Subphase. 
Recall that $s_1$ and $e_1$ are the first and the last time steps, respectively, when an agent performs an interaction in the Pólya Subphase. We get the following lemma. % (see \cref{apx:majority} for the proof).

\begin{lemma}[name=,restate=lempolya,label=lem:polya]
Let $a = A_{s_1-1}$ and $b = B_{s_1-1}$.
For any constant $\beta > 0$ there exists a constant $\alpha$ such that if $a > \alpha \cdot \log n$ and $a > (1+\beta) \cdot b$ then $A_{e_1} - B_{e_1} = \LDAUOmega{n}$ with probability at least $1 - n^{-(c+2)}$.
\end{lemma}

%\paragraph{Cancellation Subphase}
Next we consider the Cancellation Subphase.
The goal is to remove any occurrence of the minority opinion.
Whenever an agent with opinion $A$ interacts with another agent with opinion $B$, both agents become undecided.
Formally, we show the following lemma.
\begin{lemma}[name=,restate=lemcancellation,label=lem:cancellation]
If $A_{s_2-1} - B_{s_2-1} = \LDAUOmega{n}$ then
$A_{e_2} = \LDAUOmega{n}$ and $B_{e_2}=0$ with probability at least $1-n^{-(c+2)}$.
\end{lemma}

%\paragraph{Broadcasting Subphase}
Finally we consider the Broadcasting Subphase.
The goal %of the Broadcasting Subphase 
is to spread the (unique) remaining opinion to all other agents.
Whenever an undecided agent $u$ interacts with another agent $v$ that has an opinion, agent $u$ adopts the opinion of agent $v$.
This leads to a configuration where every agent has the majority opinion \whp.
Formally, we show the following lemma.
\begin{lemma}[name=,restate=lembroadcasting,label=lem:broadcasting]
If $A_{e_2} = \LDAUOmega{n}$ and $B_{e_2}=0$, then $A_{e_3} = n$ and $B_{e_3}=0$ with probability at least $1-n^{-(c+2)}$.
\end{lemma}

%\paragraph{Putting Everything Together}
\noindent We have now everything we need to prove \cref{pro:adaptive-majority}.
%Here, we sketch the proof. 
\begin{proof}[Proof of \cref{pro:adaptive-majority}]
We assume the configuration at time $t_1$ is a homogeneous launching configuration.
From \cref{pro:maintenance} it follows that all configurations in the time interval $[t_1,t_1+ t_2]$ for some $t_2 = \Theta(n \cdot w \cdot \log n)$ are synchronous with probability at least $1 - n^{-(c+1)}$.
This means that the three subphases are strictly separated as explained above.
It therefore follows, each with probability at least $1-n^{-(c+2)}$,
\begin{itemize}[nosep]
\item from \cref{lem:polya} that after the Pólya Subphase no agent is undecided,
\item from \cref{lem:cancellation} that after the Cancellation Subphase no agent has opinion $B$, and
\item from \cref{lem:broadcasting} that after the Broadcasting Subphase all agents have opinion $A$.
\end{itemize}
Once all agents have opinion $A$, this becomes the output when the agents enter $I_{gather}$.
Together, this shows that all agents output the majority opinion after $\Theta(n \cdot w \cdot \log{n} )$ interactions with probability at least $1-n^{-c}$.
%The full proof can be found in \cref{apx:majority}.
\end{proof}

\noindent With \cref{pro:adaptive-majority} we can now prove \cref{thm:adaptive-majority} as follows.

\begin{proof}[Proof of \cref{thm:adaptive-majority}]
We first show recovery.
Note that we do not (yet) consider input changes.
Fix a time $t_1$ and assume the agents are in an arbitrary configuration at time $t_1$.
From \cref{thm:loosely_stabilizing} it follows the clocks enter a synchronous configuration within $\ldauOmicron{n \log n}$ interactions and stay in synchronous configurations for $\poly n$ time \whp.

Fix a synchronized phase $i < \poly n$.
It follows from \cref{pro:adaptive-majority} that all agents enter a correct configuration at the end of phase $i$ with probability at least $1-n^{-c}$.
(Recall that in a correct configuration all agents have to output the majority opinion if there is a sufficiently large bias. Without a bias, any output constitutes a correct configuration.)

From the guarantees of the phase clock it follows that the first synchronized phase starts within $O(n \log n)$ time after time $t_1$ \whp.
This shows a convergence time of $O(n \log n)$.
From a union bound over at most $n^{c-1}$ phases it follows that the protocol is in a correct configuration for $\poly n$ interactions \whp.
This shows a holding time of $\poly n$.
Together, this concludes the proof.
\end{proof}
 
\paragraph{Improving the Bound}
In order to show-case the simplicity of the application of our phase clock, we have presented a simplistic protocol, where we assumed a constant factor bias towards the majority opinion.
We now show how to obtain a tighter result:
we replace the Cancellation Subphase and the Broadcasting Subphase (lines 6 to 9 in \cref{alg:adaptive-majority}) with the \emph{undecided state dynamics} introduced in \cite{DBLP:journals/dc/AngluinAE08}.

The undecided state dynamics is defined as follows.
Each agent is in one of three states, $x$, $y$, or $b$.
If an agent $u$ with opinion $x$ ($y$, resp.) interacts with another agent $v$ with opinion $y$ ($x$, resp.), agent $u$ enters the \emph{blank} state $b$.
If an agent $u$ in state $b$ interacts with another agent $v$ in state $x$ or $y$, $u$ adopts $v$'s state.
When adopted to our problem, the values $A$, $B$, $U$ of variable $\Opinion u$ at agent $u$ translate directly to states $x$, $y$, and $b$, respectively.
The undecided state dynamics converge to the correct majority \whp even if the multiplicative bias is much smaller than a constant fraction.
Formally, we show the following statement.
\begin{observation}[name=,restate=proImprovedBounds,label=pro:improved-bounds]
If we use the undecided state dynamics, \cref{pro:adaptive-majority} also holds for $\alpha = \LDAUOmega{ \beta^{-2} }$
provided that $\beta = \LDAUOmega{n^{-1/4+\epsilon}}$.
\end{observation}
This means that we can solve the adaptive majority problem with a multiplicative bias of $1 + \beta = 1 + 1/\log n = 1 + \Ldauomicron{1}$ and asymptotically at least $\LDAUOmega{\log^2 n}$ many agents (assuming sufficiently large constants).
Hence we match the results of \cite{DBLP:conf/podc/Alistarh0U21}.

\paragraph{Robustness Against Input Changes}

We finally investigate the effect of \emph{input changes}, introduced in \cite{DBLP:conf/podc/Alistarh0U21} as \emph{leak transitions}.
Input changes affect the $\Input{u}$ variable of an agent $u$: they convert an agent $u$ with majority input $\Input u = A$ to an agent $u$ with minority input $\Input u = B$.
(Recall that we assume w.l.o.g.\ that $A$ is the majority and $B$ is the minority opinion.)
Our main observation is the following.
Let $r$ be the rate of input changes such that in each interaction an input changes with probability $r$.
Assume that $r$ bounded and a sufficiently large bias $\beta$ towards the majority opinion is present.
Then we show the following result.

\begin{proposition}[name=,restate=proInputChanges,label=pro:input-changes]
Assume that at time $t_1$ he clocks are in a homogeneous launching configuration and we have $A_{t_1} \geq \alpha \cdot\log n$ and $A_{t_1} \geq \left(1+\beta\right)\cdot B_{t_1}$.
Assume that the inputs change at rate $r$.
If
\begin{align*}
\alpha &= \LDAUOmega{\beta^{-2}} %\left(\beta  - 1/3 \cdot n^{-1/4 + \epsilon} \right)^{-2}}
\text{, } \quad \beta = \LDAUOmega{n^{-1/4+\epsilon}} \text{, \quad and } \quad r \leq \frac{\beta \cdot \alpha}{n}
\end{align*}
%
%
%r * n log n input changes pro phase
%
%A  > (1+ beta) B + 2r * n log n 
%
%r  < beta alpha / n
%
then there exists a $t_2 = \Theta(n \cdot w \cdot \log n)$
such that all agents output $A$ in configuration $C_{t_1 + t_2}$
with probability $1 - n^{-c}$.

%\begin{enumerate}[nosep]
%\item If $r \leq 1/n$ then the improved bounds from %\cref{pro:improved-bounds} hold.
%\item If $r \leq ... $ and all agents have input either $A$ or $B$, then %\cref{pro:adaptive-majority} holds for $\beta \geq ...$.
%\end{enumerate}
\end{proposition}

\begin{proof}[Proof Sketch]
The main idea in the proof of \cref{pro:input-changes} is that we bound the number of input changes in $[t_1, t_2]$ by a simple application of Chernoff bounds.
Intuitively, if this number of input changes is at most a constant factor of the bias, we are done. 
Hence the result follows from the observation that the input configuration does not change quickly enough to turn over the bias: 
the output of the agents at the end of the phase, even though its calculation is based on an \emph{old} configuration from the beginning of the phase, is still correct. The formal statement then follows from the previous analysis without input changes.
\end{proof}

Finally, observe that also under presence of input changes we match the results of \cite{DBLP:conf/podc/Alistarh0U21} for $r = 1/n$, a multiplicative bias of $1 + \beta = 1 + 1/\log n $ and $\LDAUOmega{\log^2 n}$ many agents ($\alpha \geq \log n$).

%The above result shows that the bias at the end of the Pólya Subphase is \emph{large enough}, even if input changes may occur at rate $\LDAUOmicron{1/n}$.

\bibliography{references}

\begin{thebibliography}{10}

\bibitem{DBLP:conf/soda/AlistarhAG18}
Dan Alistarh, James Aspnes, and Rati Gelashvili.
\newblock Space-optimal majority in population protocols.
\newblock In {\em Proceedings of the Twenty-Ninth Annual {ACM-SIAM} Symposium
  on Discrete Algorithms, {SODA}}, pages 2221--2239. {SIAM}, 2018.
\newblock \href {https://doi.org/10.1137/1.9781611975031.144}
  {\path{doi:10.1137/1.9781611975031.144}}.

\bibitem{DBLP:conf/dna/AlistarhDKSU17}
Dan Alistarh, Bartlomiej Dudek, Adrian Kosowski, David Soloveichik, and
  Przemyslaw Uznanski.
\newblock Robust detection in leak-prone population protocols.
\newblock In {\em {DNA} Computing and Molecular Programming - 23rd
  International Conference, {DNA}}, volume 10467 of {\em Lecture Notes in
  Computer Science}, pages 155--171. Springer, 2017.
\newblock \href {https://doi.org/10.1007/978-3-319-66799-7\_11}
  {\path{doi:10.1007/978-3-319-66799-7\_11}}.

\bibitem{DBLP:conf/podc/Alistarh0U21}
Dan Alistarh, Martin T{\"{o}}pfer, and Przemyslaw Uznanski.
\newblock Comparison dynamics in population protocols.
\newblock In Avery Miller, Keren Censor{-}Hillel, and Janne~H. Korhonen,
  editors, {\em {PODC} '21: {ACM} Symposium on Principles of Distributed
  Computing, Virtual Event, Italy, July 26-30, 2021}, pages 55--65. {ACM},
  2021.
\newblock \href {https://doi.org/10.1145/3465084.3467915}
  {\path{doi:10.1145/3465084.3467915}}.

\bibitem{DBLP:conf/opodis/AmirAL20}
Talley Amir, James Aspnes, and John Lazarsfeld.
\newblock Approximate majority with catalytic inputs.
\newblock In Quentin Bramas, Rotem Oshman, and Paolo Romano, editors, {\em 24th
  International Conference on Principles of Distributed Systems, {OPODIS} 2020,
  December 14-16, 2020, Strasbourg, France (Virtual Conference)}, volume 184 of
  {\em LIPIcs}, pages 19:1--19:16. Schloss Dagstuhl - Leibniz-Zentrum f{\"{u}}r
  Informatik, 2020.
\newblock \href {https://doi.org/10.4230/LIPIcs.OPODIS.2020.19}
  {\path{doi:10.4230/LIPIcs.OPODIS.2020.19}}.

\bibitem{DBLP:journals/dc/AngluinADFP06}
Dana Angluin, James Aspnes, Zo{\"{e}} Diamadi, Michael~J. Fischer, and
  Ren{\'{e}} Peralta.
\newblock Computation in networks of passively mobile finite-state sensors.
\newblock {\em Distributed Comput.}, 18(4):235--253, 2006.
\newblock \href {https://doi.org/10.1007/s00446-005-0138-3}
  {\path{doi:10.1007/s00446-005-0138-3}}.

\bibitem{DBLP:conf/podc/AngluinAE06}
Dana Angluin, James Aspnes, and David Eisenstat.
\newblock Stably computable predicates are semilinear.
\newblock In {\em Proceedings of the Twenty-Fifth Annual {ACM} Symposium on
  Principles of Distributed Computing, {PODC}}, pages 292--299. {ACM}, 2006.
\newblock \href {https://doi.org/10.1145/1146381.1146425}
  {\path{doi:10.1145/1146381.1146425}}.

\bibitem{DBLP:journals/dc/AngluinAE08a}
Dana Angluin, James Aspnes, and David Eisenstat.
\newblock Fast computation by population protocols with a leader.
\newblock {\em Distributed Comput.}, 21(3):183--199, 2008.
\newblock \href {https://doi.org/10.1007/s00446-008-0067-z}
  {\path{doi:10.1007/s00446-008-0067-z}}.

\bibitem{DBLP:journals/dc/AngluinAE08}
Dana Angluin, James Aspnes, and David Eisenstat.
\newblock A simple population protocol for fast robust approximate majority.
\newblock {\em Distributed Comput.}, 21(2):87--102, 2008.
\newblock \href {https://doi.org/10.1007/s00446-008-0059-z}
  {\path{doi:10.1007/s00446-008-0059-z}}.

\bibitem{DBLP:journals/dc/AngluinAER07}
Dana Angluin, James Aspnes, David Eisenstat, and Eric Ruppert.
\newblock The computational power of population protocols.
\newblock {\em Distributed Comput.}, 20(4):279--304, 2007.
\newblock \href {https://doi.org/10.1007/s00446-007-0040-2}
  {\path{doi:10.1007/s00446-007-0040-2}}.

\bibitem{DBLP:journals/jcss/Aspnes21}
James Aspnes.
\newblock Clocked population protocols.
\newblock {\em J. Comput. Syst. Sci.}, 121:34--48, 2021.
\newblock \href {https://doi.org/10.1016/j.jcss.2021.05.001}
  {\path{doi:10.1016/j.jcss.2021.05.001}}.

\bibitem{DBLP:journals/corr/abs-2103-10366}
Gregor Bankhamer, Petra Berenbrink, Felix Biermeier, Robert Els{\"{a}}sser,
  Hamed Hosseinpour, Dominik Kaaser, and Peter Kling.
\newblock Fast consensus via the unconstrained undecided state dynamics, 2022.
\newblock SODA 2022, to appear.
\newblock \href {http://arxiv.org/abs/2103.10366} {\path{arXiv:2103.10366}}.

\bibitem{DBLP:conf/podc/BankhamerEKK20}
Gregor Bankhamer, Robert Els{\"{a}}sser, Dominik Kaaser, and Matjaz Krnc.
\newblock Positive aging admits fast asynchronous plurality consensus.
\newblock In {\em {PODC} '20: {ACM} Symposium on Principles of Distributed
  Computing}, pages 385--394. {ACM}, 2020.
\newblock \href {https://doi.org/10.1145/3382734.3406506}
  {\path{doi:10.1145/3382734.3406506}}.

\bibitem{BKKP20}
Stav Ben-Nun, Tsvi Kopelowitz, Matan Kraus, and Ely Porat.
\newblock An $o(\log^{3/2}n)$ parallel time population protocol for majority
  with $o(\log n)$ states.
\newblock In {\em Proceedings of the {ACM} Symposium on Principles of
  Distributed Computing, {PODC} 2020, Virtual Event, Italy, August 3-7, 2020},
  page to appear, 2020.

\bibitem{DBLP:journals/corr/abs-1805-04586}
Petra Berenbrink, Robert Els{\"{a}}sser, Tom Friedetzky, Dominik Kaaser, Peter
  Kling, and Tomasz Radzik.
\newblock Time-space trade-offs in population protocols for the majority
  problem.
\newblock {\em Distributed Computing}, 2020.
\newblock \href {https://doi.org/10.1007/s00446-020-00385-0}
  {\path{doi:10.1007/s00446-020-00385-0}}.

\bibitem{DBLP:journals/dc/BerenbrinkEFKKR21}
Petra Berenbrink, Robert Els{\"{a}}sser, Tom Friedetzky, Dominik Kaaser, Peter
  Kling, and Tomasz Radzik.
\newblock Time-space trade-offs in population protocols for the majority
  problem.
\newblock {\em Distributed Comput.}, 34(2):91--111, 2021.
\newblock \href {https://doi.org/10.1007/s00446-020-00385-0}
  {\path{doi:10.1007/s00446-020-00385-0}}.

\bibitem{DBLP:conf/stoc/BerenbrinkGK20}
Petra Berenbrink, George Giakkoupis, and Peter Kling.
\newblock Optimal time and space leader election in population protocols.
\newblock In {\em Proccedings of the 52nd Annual {ACM} {SIGACT} Symposium on
  Theory of Computing, {STOC}}, pages 119--129. {ACM}, 2020.
\newblock \href {https://doi.org/10.1145/3357713.3384312}
  {\path{doi:10.1145/3357713.3384312}}.

\bibitem{DBLP:conf/podc/BurmanCCDNSX21}
Janna Burman, Ho{-}Lin Chen, Hsueh{-}Ping Chen, David Doty, Thomas Nowak,
  Eric~E. Severson, and Chuan Xu.
\newblock Time-optimal self-stabilizing leader election in population
  protocols.
\newblock In Avery Miller, Keren Censor{-}Hillel, and Janne~H. Korhonen,
  editors, {\em {PODC} '21: {ACM} Symposium on Principles of Distributed
  Computing, Virtual Event, Italy, July 26-30, 2021}, pages 33--44. {ACM},
  2021.
\newblock \href {https://doi.org/10.1145/3465084.3467898}
  {\path{doi:10.1145/3465084.3467898}}.

\bibitem{DBLP:journals/nc/CondonHKM20}
Anne Condon, Monir Hajiaghayi, David~G. Kirkpatrick, and J{\'{a}}n Manuch.
\newblock Approximate majority analyses using tri-molecular chemical reaction
  networks.
\newblock {\em Nat. Comput.}, 19(1):249--270, 2020.
\newblock \href {https://doi.org/10.1007/s11047-019-09756-4}
  {\path{doi:10.1007/s11047-019-09756-4}}.

\bibitem{10.1145/1017460.1017463}
Shlomi Dolev and Jennifer~L. Welch.
\newblock Self-stabilizing clock synchronization in the presence of byzantine
  faults.
\newblock {\em J. ACM}, 51(5):780–799, September 2004.
\newblock \href {https://doi.org/10.1145/1017460.1017463}
  {\path{doi:10.1145/1017460.1017463}}.

\bibitem{doty2021time}
David Doty, Mahsa Eftekhari, Leszek Gąsieniec, Eric Severson, Grzegorz
  Stachowiak, and Przemysław Uznański.
\newblock A time and space optimal stable population protocol solving exact
  majority, 2022.
\newblock FOCS 2022, to appear.
\newblock \href {http://arxiv.org/abs/2106.10201} {\path{arXiv:2106.10201}}.

\bibitem{DBLP:journals/dc/DotyS18}
David Doty and David Soloveichik.
\newblock Stable leader election in population protocols requires linear time.
\newblock {\em Distributed Computing}, 31(4):257--271, 2018.
\newblock \href {https://doi.org/10.1007/s00446-016-0281-z}
  {\path{doi:10.1007/s00446-016-0281-z}}.

\bibitem{DBLP:conf/stoc/DudekK18}
Bartlomiej Dudek and Adrian Kosowski.
\newblock Universal protocols for information dissemination using emergent
  signals.
\newblock In {\em Proceedings of the 50th Annual {ACM} {SIGACT} Symposium on
  Theory of Computing, {STOC}}, pages 87--99. {ACM}, 2018.
\newblock \href {https://doi.org/10.1145/3188745.3188818}
  {\path{doi:10.1145/3188745.3188818}}.

\bibitem{DBLP:conf/soda/GasieniecS18}
Leszek Gasieniec and Grzegorz Stachowiak.
\newblock Fast space optimal leader election in population protocols.
\newblock In {\em Proceedings of the Twenty-Ninth Annual {ACM-SIAM} Symposium
  on Discrete Algorithms, {SODA}}, pages 2653--2667. {SIAM}, 2018.
\newblock \href {https://doi.org/10.1137/1.9781611975031.169}
  {\path{doi:10.1137/1.9781611975031.169}}.

\bibitem{Janson17}
Svante Janson.
\newblock Tail bounds for sums of geometric and exponential variables.
\newblock {\em Statistics \& Probability Letters}, 135, 09 2017.
\newblock \href {https://doi.org/10.1016/j.spl.2017.11.017}
  {\path{doi:10.1016/j.spl.2017.11.017}}.

\bibitem{JK77}
Norman~Lloyd Johnson and Samuel Kotz.
\newblock {\em {U}rn {M}odels and {T}heir {A}pplication: {A}n {A}pproach to
  {M}odern {D}iscrete {P}robability {T}heory}.
\newblock Wiley, 1977.

\bibitem{DBLP:conf/podc/KosowskiU18}
Adrian Kosowski and Przemyslaw Uznanski.
\newblock Brief announcement: Population protocols are fast.
\newblock In {\em Proceedings of the 2018 {ACM} Symposium on Principles of
  Distributed Computing, {PODC}}, pages 475--477. {ACM}, 2018.
\newblock \href {https://doi.org/10.1145/3212734.3212788}
  {\path{doi:10.1145/3212734.3212788}}.

\bibitem{DBLP:books/daglib/0012859}
Michael Mitzenmacher and Eli Upfal.
\newblock {\em Probability and Computing: Randomized Algorithms and
  Probabilistic Analysis}.
\newblock Cambridge University Press, 2005.
\newblock \href {https://doi.org/10.1017/CBO9780511813603}
  {\path{doi:10.1017/CBO9780511813603}}.

\bibitem{DBLP:journals/tcs/SudoNYOKM12}
Yuichi Sudo, Junya Nakamura, Yukiko Yamauchi, Fukuhito Ooshita, Hirotsugu
  Kakugawa, and Toshimitsu Masuzawa.
\newblock Loosely-stabilizing leader election in a population protocol model.
\newblock {\em Theor. Comput. Sci.}, 444:100--112, 2012.
\newblock \href {https://doi.org/10.1016/j.tcs.2012.01.007}
  {\path{doi:10.1016/j.tcs.2012.01.007}}.

\bibitem{DBLP:conf/icalp/TalwarW14}
Kunal Talwar and Udi Wieder.
\newblock Balanced allocations: {A} simple proof for the heavily loaded case.
\newblock In {\em Automata, Languages, and Programming - 41st International
  Colloquium, {ICALP}}, volume 8572 of {\em Lecture Notes in Computer Science},
  pages 979--990. Springer, 2014.
\newblock \href {https://doi.org/10.1007/978-3-662-43948-7\_81}
  {\path{doi:10.1007/978-3-662-43948-7\_81}}.

\end{thebibliography}

\appendix
\section*{Appendix}
\section{Auxiliary Results}
\label{apx:auxiliary-results}

In this appendix we state a number of auxiliary results used in our analysis for completeness.
\subsection{Concentration Inequalities}
We start with classical Chernoff bounds.
\begin{theorem}[\cite{DBLP:books/daglib/0012859}, Theorem $4.4$, $4.5$]
\label{lemma:chernoff_poisson_trials}
\label{thm:chernoff-bounds}
    Let $X_1, \dots, X_n$ be independent Poisson trials with $\Pr[X_i = 1] = p_i$ and let $X = \sum X_i$ with $\Ex{X} = \mu$. Then the following Chernoff bounds hold
    for $0 < \delta \leq 1$:
    \begin{align*}
        \Pr[X > (1+\delta) \cdot\mu] &\leq e^{-\mu \cdot {\delta}^2 / 3}, \quad \text{and} \\
        \Pr[X < (1-\delta) \cdot\mu] &\leq e^{-\mu \cdot {\delta}^2 / 2}.
    \end{align*}
\end{theorem}

\noindent Next we consider tail bounds for sums of geometrically distributed random variables.
\begin{theorem}[\cite{Janson17}, Theorem $2.1$]
\label{lem:janson}
Let $X = \sum_{i=1}^{n} X_i$ where $X_i, i = 1,\dots,n$, are independent geometric random variables with $X_i \sim Geo(p_i)$ for $p_i \in (0,1]$.
For any $\lambda \geq 1$,
\begin{equation*}
    \Prob{X \geq \lambda \cdot \Ex{X}} \leq \exp({- \min_{i}\{p_i\} \cdot \Ex{X} \cdot (\lambda - 1 - \ln{\lambda})}).
\end{equation*}
\end{theorem}

\noindent The following theorem considers Pólya urns. 
Recall that Pólya urns are defined as follows.
Initially, the urn contains $a$ red balls and $b$ blue balls.
In each step, a ball is drawn uniformly at random from the urn.
The ball's color is observed, and it is returned into the urn along with an additional ball of the same color.
The \emph{Pólya-Eggenberger distribution}, denoted by $\PE(a,b,n)$, describes the \emph{total number} of red balls that are contained in the urn after $n$ steps.
Note that in some related works the distribution describes the number of additional red balls (instead of the total number).
\begin{theorem}[\cite{DBLP:conf/podc/BankhamerEKK20}, Theorem 1]
\label{thm:polya_bound_total_balls}
Let $A \sim \PE(a,b,n- (a+b))$,  $\mu = (a/(a+b)) n$ and $a+b \geq 1$. Then, for any $\delta$ with $0 < \delta < \sqrt{a}$ and some small constant $1 > \polyaConstSmall>0$ it holds that 
\begin{align*}
    \Pr \Big(A < \mu - \sqrt{a} \cdot \frac{n}{a+b} \cdot  \delta\Big) &< 4 \exp(-\polyaConstSmall \cdot \delta^2) \\
    \Pr \Big(A > \mu + \sqrt{a} \cdot \frac{n}{a+b} \cdot  \delta \Big) &< 4 \exp(-\polyaConstSmall \cdot \delta^2)
\end{align*}
\end{theorem}

\noindent Finally, we state a result regarding the undecided dynamics introduced in \cite{DBLP:journals/dc/AngluinAE08}.
\begin{theorem}[\cite{DBLP:journals/dc/AngluinAE08}, Theorem 3]
\label{thm:undecided-dynamics}
Let $\epsilon > 0$. If the difference between the initial majority and initial minority populations is $\Omega(n^{3/4+\epsilon})$ and there is exactly one active agent, then with high probability,
the epidemic-triggered approximate majority protocol converges to the initial majority value.
\end{theorem}

\subsection{Simple Analysis of the One-Way Epidemic}
In this appendix we give a simple proof of the one-way epidemic that is based on the analysis in \cite{DBLP:journals/dc/AngluinAE08a}.
\label{apx:one-way-epidemic}
Assume that at time $0$ one agent is infected.
Let $T_i$ denote the random variable counting the number of interactions until one of the $n-i$ susceptible agents initiate an interaction with one of the $i$ informed agents.
$T_i$ is geometrically distributed with success probability $p_i = {i \cdot (n-i)}/{n^2}$.
We calculate the expected number of interactions until all agents are infected, i.e., until there are $n-1$ successes.
First, a simple calculation yields a lower and upper bound on the expected value of the sum of $T_i$'s.
\begin{align*}
    & \Ex{ \sum_{i=1}^{n-1} T_i } 
    %= \sum_{i=1}^{n-1} \Ex{T_i} 
    = \sum_{i=1}^{n-1} \frac{1}{p_i} 
    = n^2 \sum_{i=1}^{n-1} \frac{1}{i \cdot (n-i)} 
    \leq 2 \cdot n^2 \sum_{i=1}^{\ceil{\frac{n-1}{2}}} \frac{1}{i \cdot (n-\ceil{\frac{n-1}{2}})} 
    \leq 4 \cdot n \cdot \log n\\
    & \Ex{ \sum_{i=1}^{n-1} T_i } 
    = n^2 \sum_{i=1}^{n-1} \frac{1}{i \cdot (n-i)} 
    \geq n \sum_{i=1}^{n-1} \frac{1}{i}
    \geq n \cdot \ln(n)\\
\end{align*}
Since all trials are independent, we can also calculate an upper bound.    
An application of \cref{lem:janson} with $\lambda = 9(c+4)/(4\log(e))$ and $\min_{i}\{ p_i \} = (n-1)/n^2$ yields (for sufficiently large $n$)
\begin{align*}
    \MoveEqLeft \Prob{\sum_{i=1}^{n-1} T_i > 6 (c+4) \cdot n \cdot \ln n}
    \leq \Prob{\sum_{i=1}^{n-1} T_i  > \lambda \cdot \Ex{ \sum_{i=1}^{n-1} T_i }}\\
    & < \exp\left(-\frac{n-1}{n^2} \cdot n \cdot \ln n \cdot \left(\lambda-1-\ln \lambda \right)\right)
    \leq n^{- \left( 1 - \frac{1}{n} \right) \left(\lambda-1-\ln \lambda \right) }
    \leq n^{-(5+c)}.
\end{align*}
 %where we used the bounds on the expected value and the simplification $c+4 > \ln (3c/2+6)$.

\section{Omitted Proofs for the Loosely-Stabilizing Phase Clocks}

In this appendix we give additional details and the full proofs for the loosely-stabilizing phase clock that have been omitted from \cref{sec:maintenance,sec:recovery}.

\lemmaRecoveryOne*

\begin{proof}
    Recall the main idea of the proof:
    If there are not too many agents in $I_{gather}$, the reset rule prevents agents from reaching the end of $I_{gather}$.
    Agents may still enter $I_{launch}$ by hopping, but if no agent enters state $0$, eventually there is no agent left in state $0$ to hop onto.
    This applies to all other states of $I_{launch}$ as well.
    Eventually, there are no agents left in $I_{launch}$ to hop onto.
    Then, no agent leaves $I_{gather}$ until a reset is sufficiently unlikely such that an agent can reach the end of the interval.
    We use the rather arbitrary threshold of $0.9n$ agents in the definition of an \emph{almost homogeneous gathering configuration}.

    Formally, we divide time into \emph{stages} of $2\tau n$ consecutive interactions as follows. 
    W.l.o.g.\ assume that $t=0$.
    For $i \geq 0$, we define the $i$'th stage as $\stage{i} = \intcc{i \cdot 2 \tau n, (i+1) \cdot 2\tau n - 1}$.
    Our proof is based on a case distinction over the following two predicates:
    \begin{itemize}
        \item $\FewInGather(\stage{i})$ holds if for all configurations $C_{t}$, $t \in \stage{i}$ there are less than $0.9 \cdot n$ agents in $I_{gather}$.
        \item $\NoBroadcast(\stage{i})$ holds if none of the interactions in $\stage{i}$ is of the type $(|Q|-1,q_2) \rightarrow (0,q_2), q_2 \in I_{gather}$, i.e., no agent enters $I_{launch}$ without hopping.
    \end{itemize}
    
    \medskip
    
\noindent We divide the proof into three parts. We show that the following statements each hold with probability at least $1-1/(9n)$.
    \begin{enumerate}
        \item $\exists i \in \intcc{0,5}: \NoBroadcast(\stage{i}) \land \NoBroadcast(\stage{i+1})$.
        \item for any fixed $i \geq 0: \NoBroadcast(\stage{i}) \land \NoBroadcast(\stage{i+1}) \land \FewInGather(\stage{i+1}) \implies \exists j \in \LdauOmicron{w}: \NoBroadcast(\stage{i+j}) \land \NoBroadcast(\stage{i+j+1}) \land \overline{\FewInGather}(\stage{i+j+1})$ 
        \item for any fixed $i \geq 0: \NoBroadcast(\stage{i}) \land \NoBroadcast(\stage{i+1}) \land \overline{\FewInGather}(\stage{i+1}) \implies \exists t \in \stage{i+1}: C_{t}$ is an almost homogeneous gathering configuration.
    \end{enumerate}
        
    Throughout the proof we need a bound on the number of initiated interactions per agent during a stage.
    Let $X_u(i)$ denote the number of interactions agent $u$ initiates in $\stage{i}$.
    From \cref{lem:moving} it follows with $\delta=1$ that $X_u(i) > \tau$ with probability at least $1-n^{-36\cdot(c+4)}$
    and that $X_u(i) < 4 \cdot \tau$ with probability at least $1-n^{-24\cdot(c+4)}$.
    By a union bound, we get for any fixed $\stage{i}$ that
    \begin{align}
        \label{fact:moving}
        \forall u \in V: \tau < X_u(i) < 4 \cdot \tau \text{ with probability at least } 1-n^{-36\cdot(c+4)+1}-n^{-24\cdot(c+4)+1}.
    \end{align}
    
    \paragraph{Statement 1}
    We first show that \whp there exists an $i \in \intcc{0,3}$ such that $\FewInGather(\stage{i}) \lor \NoBroadcast(\stage{1})$.
    Then, we repeat the argument for $i \in \intcc{1,4}$.
    Finally, we show that $\FewInGather(\stage{i}) \implies \NoBroadcast(\stage{i+1}) \land \NoBroadcast(\stage{i+1})$.
    
    For the sake of the argument, assume that $\overline{\FewInGather}(\stage{0}) \land \overline{\FewInGather}(\stage{1}) \land \overline{\FewInGather}(\stage{2}) \land \overline{\FewInGather}(\stage{3}) \land \overline{\NoBroadcast}(\stage{1})$ holds. We show that \whp this can not happen. The case for $\overline{\FewInGather}(\stage{1}) \land \overline{\FewInGather}(\stage{2}) \land \overline{\FewInGather}(\stage{3}) \land \overline{\FewInGather}(\stage{4}) \land \overline{\NoBroadcast}(\stage{2})$ is very similar.
    
    From $\overline{\FewInGather}(\stage{i})$ it follows that there exists a configuration $C_{t_i} (0\le i\le 3)$ $, t_i \in \stage{i}$ where there are $\ge 0.9\cdot n$ agents in $I_{gather}$.
    Any agent that leaves $I_{gather}$ during $\stage{i}$ does not return to $I_{gather}$ within the next three stages \whp.
    This follows from \cref{fact:moving} and $|I_{work}| > 16 \tau$.
    Therefore, no more than $0.1n$ agents may leave $I_{gather}$ during $\stage{i} \cup \stage{i+1} \cup \stage{i+2}$.
    Thus, \whp there are always at least $0.8\cdot n$ agents in $I_{gather}$ for all $t \in \stage{0} \cup \stage{1} \cup \stage{2}$.
    
    Let us fix configuration $C_{t'}\in \stage{1}$ with an agent $v_1$ in state $0$. 
    By our assumption such a configuration exists, otherwise $\NoBroadcast(\stage{1})$ holds.
    From above it follows that at least $0.8n$ agents are in $I_{gather}$ in $C_{t'}$.
    In the following, we will show that \whp at least $0.1\cdot n$ agents leave $I_{gather}$ in $\intcc{t',t'+\tau/4 \cdot n}$.
    Note that this contradicts our assumption.
    
    The proof that at least $0.1\cdot n$ agents leave $I_{gather}$ is similar to the one of \cref{lem:gather->launch}. 
    The main difference is that some agents are already in $I_{launch}$ and can leave $I_{launch}$ during that time.
    Fix an arbitrary set of agents that contains $0.7n-1$ agents in $I_{gather}$ and $v_1$.
    We label these agents as \emph{relevant}.
    An interaction between two relevant agents occurs with probability at least $p=0.7^2-0.21/(n-1)$.
    Let $\tilde{t} = n\tau/4$.
    By \cref{thm:chernoff-bounds}, with probability at least $1-n^{-3(c+4)}$ there are at least $\tilde{t}/2$ such interactions within the time interval $\intcc{t,\tilde{t}}$.
    From \cref{lem:one-way-epidemic} (with $\lambda = 9(c+4)/(8\log(e))$) it follows (for sufficiently large $c$) that with probability at least $1-n^{-(c/2+2)}$ all relevant agents enter $I_{launch}$ in $\tilde{t}/2$ interactions.
    The remainder of the proof (that \whp these agents do not leave $I_{launch}$ during $\intcc{0,\tilde{t}}$) is identical to that of \cref{lem:gather->launch:2} of \cref{lem:gather->launch}.
    This holds regardless of the choice of the relevant agents.
    
    In conclusion, our assumption must be false.
    Thus, no agent leaves $I_{gather}$ on its own in $\stage{2} \cup \stage{3}$ or in at least one of the five stages there are always less than $0.9\cdot n$ agents in $I_{gather}$. 
    
    We complete the proof by showing the following for any fixed $i \geq 0$:
    \begin{align}
        \label{fact:resetting}
        \FewInGather(\stage{i}) \implies \NoBroadcast(\stage{i+1}) \land \NoBroadcast(\stage{i+2}) \text{ with probability at least }1-n^{-113(c+4)+1}.
    \end{align}
    We show that if $\FewInGather(\stage{i})$ holds, all agents either hop or reset and that any agent such agent does not reach the end of $I_{gather}$.
    Recall that agents in $I_{gather}$ reset to the beginning of $I_{gather}$ when they interact with agents in $I_{work}$, and they hop into $I_{launch}$ when they interact with agents in $I_{launch}$.
    Consider $\stage{i}$ where $\FewInGather(\stage{i})$ holds, i.e., there are less than $0.9\cdot n$ agents in $I_{gather}$.
    It follows that, with probability at least $0.1$ per initiated interaction, an agent in $I_{gather}$ resets or leaves $I_{gather}$.
    We will call this a \emph{success}.
    It is easy to see that each agent has \whp at least one success in $\stage{i}$.
    From \cref{fact:moving} it follows that each agent initiates at least $\tau$ interactions in $\stage{i}$ with probability at least $1-n^{-12\cdot(c+4)+1}-n^{-8\cdot(c+4)+1}$.
    For any agent in $I_{gather}$, the probability that it has at least one success in $\tau$ interactions is at least $1-(1-0.1)^{\tau} = 1-n^{-113(c+4)}$.
    It follows from the union bound that all agents that are initially in $I_{gather}$ during $\stage{i}$ have a success with probability at least $1-3n^{-8(c+4)+1}$.
    It remains to show that an agent that resets is unlikely to leave $I_{gather}$ without hopping.
    Note that it is even more unlikely for agents which hop into $I_{launch}$ since they have to walk around the whole clock. 
    From \cref{fact:moving} it follows that each agent initiates at most $12 \cdot \tau$ interactions in three stages \whp.
    Since $|I_{gather}| > 12 \cdot \tau$, no successful agent leaves $I_{gather}$ on its own in $\stage{i}$, $\stage{i+1}$, and $\stage{i+2}$.
    This implies $\NoBroadcast(\stage{i+1}) \land \NoBroadcast(\stage{i+2})$.
    
    As we have shown above, \whp $\NoBroadcast(\stage{2}) \land \NoBroadcast(\stage{3})$ holds or there exists $i \leq 4$ such that $\FewInGather(\stage{i})$ holds.
    Statement 1 then follows from \cref{fact:resetting}.
    
    \paragraph{Statement 2}
    We begin with the following observation.
    Let $L$ be the set of agents at time $i\cdot 2\tau n$ in $I_{launch}$.
    By \cref{lem:moving} (with $\delta=1$) we know that each agent initiates at least $\tau$ interactions in $\stage{i}$ with probability at least $1-n^{-12\cdot(c+4)}$.
    Thus, if $\NoBroadcast(\stage{i})$ holds, no agent initially in $I_{launch}$ stays in that interval until the beginning of the next phase.
    We will show that this also holds for all other agents that hop.
    Consider an agent $u \in V\setminus L$ that hops on agent $v \in L$ during $\stage{i}$.
    The agents $u$ and $v$ start from the same state once $u$ has hopped, thus they have the same probability to leave $I_{launch}$.
    Since all agents in $L$ leave $I_{launch}$ \whp, this holds for $u$ as well.
    This argument extends to all agents $u \in V\setminus L$ via induction.
    
    We use this fact to prove Statement 2 as follows.
    $\NoBroadcast(\stage{i})$ holds, thus there is no agent in $I_{launch}$ in $\stage{i+1}$ \whp.
    Thus, an agent must leave $I_{gather}$ before agents can hop.
    If $\FewInGather(\stage{i+1})$ holds, the agents reset and do not leave $I_{gather}$ (see \cref{fact:resetting}).
    Therefore, in $\NoBroadcast(\stage{i+2}) \land \NoBroadcast(\stage{i+3})$ holds and no agent leaves during $\stage{i+2} \cup \stage{i+3}$ at all.

    We can repeat the argument until we reach a configuration for some stage $i+j$ where the number of agents in $I_{gather}$ reaches $0.9\cdot n$ agents.
    We have seen in the proof of \cref{lem:launch->gather:1} of \cref{lem:launch->gather} that all agents reach $I_{gather}$ in $\LdauOmicron{n \cdot w \cdot \tau}$ interactions, i.e., it takes $\LdauOmicron{w}$ stages, until sufficiently many agents have returned.
    Therefore, the statement holds for sufficiently large $c$.
    
    \paragraph{Statement 3}
    $\NoBroadcast({\stage{i}})$ implies that \whp no agent is in $I_{gather}$ at the beginning of $\stage{i+1}$ (see Statement 2).
    Thus, agents cannot hop.
    $\NoBroadcast({\stage{i+1}})$ then implies that no agent leaves $I_{gather}$ at all during $\stage{i+1}$.
    Therefore, there is no agent in $I_{launch}$ for all $t \in \stage{i+1}$.
    $\overline{\FewInGather}(\stage{i+1})$ implies the existence of a configuration $C_{t}$ with at least $0.9n$ agents in $I_{gather}$.
    Then, $C_{t}$ is an almost homogeneous gathering configuration.
\end{proof}

\lemmaRecoveryTwo*

\begin{proof}
Recall that in $C_{t}$ no agent is in $I_{launch}$ and at least $0.9\cdot n$ many agents are in $I_{gather}$ by definition.
Thus, agents cannot hop until an agent enters $I_{launch}$ on its own.
If no agent enters $I_{launch}$ on its own before the last agent enters $I_{gather}$, we are in a homogeneous gathering configuration.
Otherwise, we will show that all agents leave $I_{gather}$ at most once before we reach a homogeneous gathering configuration.

Formally, let $t+t''$ be the first time after $t$ where an agent in $I_{gather}$ reaches state $0$ without hopping.
Then, no agent has left $I_{gather}$ during $\intcc{t,t+t''}$ since there is no agent in $I_{launch}$ to hop on.
Let $t' = n \cdot (|I_{launch}|+|I_{work}|)/(1-(2 \cdot \sqrt{|I_{work}|/\tau})^{-1})$.
Any agent not in $I_{gather}$ needs at most $t'$ interactions to reach $I_{gather}$ with probability at least $1-n^{-(c+4)}$ (see \cref{lem:launch->gather:1} of \cref{lem:launch->gather}).
If $t'' > t'$, the proof is complete.
Therefore, for the remainder of this proof, we assume that $t'' \leq t'$ and w.l.o.g.\ for the sake of brevity of notation that $t+t'' = 0$.
We know for $C_{t+t''}$ that there are at least $0.9n-1$ agents in $I_{gather}$, one agent is in state $0$ and the remaining agents are in $I_{work}$.
W.l.o.g.\ we assume that $v_1$ is in state $0$.

We split the remainder of the proof into two parts.
\begin{enumerate}
    \item \Whp, after $\tau/4 \cdot n$ interactions there are at least $0.9\cdot n$ agents in $I_{launch}$ and no agent is in a state $ > |I_{launch}| + |I_{work}| + \tau$.
    \item Let $C_{t}$ be a configuration where there are at least $0.9\cdot n$ agents in $I_{launch}$ and no agent in a state $ \geq |I_{launch}| + |I_{work}| + \tau$. Then \whp, no agent leaves $I_{gather}$ before all agents have reached $I_{gather}$.
\end{enumerate}
\paragraph{Statement 1}
    We label $v_1$ and all agents that are in $I_{gather}$ at time $0$ as \emph{relevant}.
    An interaction between two relevant agents occurs with probability at least $p=0.81-0.09/(n-1)$.
    Let $\tilde{t} = n\tau/4$.
    By \cref{thm:chernoff-bounds}, \whp there are at least $\tilde{t}/2$ such interactions within the time interval $\intcc{t,\tilde{t}}$.
    From \cref{lem:one-way-epidemic} (with $\lambda = 9(c+4)/(8\log(e))$) it follows (for sufficiently large $c$) that with probability at least $1-n^{-(c/2+2)}$ all relevant agents enter $I_{launch}$ in $\tilde{t}/2$ interactions.
    The remainder of the proof (that \whp these agents do not leave $I_{launch}$ during $\intcc{0,\tilde{t}}$) is identical to that of \cref{lem:gather->launch:2} of \cref{lem:gather->launch}.
    
    Now we consider the remaining not relevant agents.
    Fix such an agent $u$.
    If $u$ hops, we have shown that at time $\tilde{t}$ $u$ is in $I_{launch}$.
    Otherwise, from \cref{lem:moving} it follows with $\delta=1$ that agent $u$ initiates less than $\tau$ interactions in $\intcc{0,\tilde{t}}$ with probability at least $1-n^{-2\cdot(c+4)}$.
    Therefore, $u$ does not reach a state $> |I_{launch}| + |I_{work}| + \tau$.
    The first statement follows from the union bound over the remaining agents.
    
\paragraph{Statement 2}
    In the following, we generalize the proof of \cref{lem:launch->gather:1} of \cref{lem:launch->gather}.
    Consider a configuration $C_{t}$ where there are at least $0.9\cdot n$ agents in $I_{launch}$ and no agent is in a state $ > |I_{launch}| + |I_{work}| + \tau$ and for the sake of brevity assume w.l.o.g.\ that $t=0$.
    For the analysis, we split the agents in two sets $L$ and $G$.
    Let $L \subset V$ be the set of all agents that start in $I_{launch}$ or hop during the first $2\tau n$ interactions
    and let $G = V\setminus L$.
    We will show that the agents in $L$ arrive in $I_{gather}$ closely together.
    The analysis of this part is almost identical to the original proof.
    Additionally, we will show that the agents in $G$ remain within the first $\tau$ states of $I_{gather}$ until the first agent of $L$ enters $I_{gather}$.
    The size of $I_{gather}$ is sufficiently large such that this 'head start' is not enough for these agents to reach the end of $I_{gather}$ before the last agent enters $I_{gather}$.
    
    First we consider the agents in $L$.
    Let $t_a$ be the first interaction in which an agent of $L$ enters $I_{gather}$.
    We have seen in \cref{lem:launch->gather} that \whp $t_a > 2\tau n$ and that no agent $u \in L$ is still in $I_{launch}$ after $2\tau n$ interactions.
    Furthermore, we have seen that after $t' = n \cdot (|I_{launch}|+|I_{work}|)/(1-(2 \cdot \sqrt{|I_{work}|/\tau})^{-1})$
    all agents $u \in L$ have reached $I_{gather}$.
    Let $t_b$ be the first interaction in which an agent of $L$ enters state $|I_{launch}|+|I_{work}|+|I_{gather}|-\tau$.
    We show that \whp $t_b > t'$.
    Let $X_u(t')$ denote the number of interactions agent $u \in L$ initiates before time $t'$.
    From \cref{lem:moving} it follows with $\delta = (|I_{launch}| +|I_{work}|-\tau) \cdot (1-1/(4+2\cdot \sqrt{10+w})) /  (|I_{launch}|+|I_{work}|)> 3/4$
    that $X_u(t') < |I_{work}| + |I_{gather}| - \tau$ with probability at least $1-n^{-9(1+c/4)}$.
    Thus, $\Clock{u}(t') \leq \Clock{u}(0) + X_u(t') < |I_{launch}| + |I_{work}| + |I_{gather}| - \tau - 1$ with probability at least $1-n^{-9(1+c/4)}$.
    By a union bound, this holds for all agents in $L$ with probability at least $1-n^{-(8+9c/4)}$.
    
    The remaining agents start in a state $< |I_{launch}| + |I_{work}| + \tau$.
    To leave $I_{gather}$ on their own, they must increase their counter at least $|I_{gather}|-\tau$ times without resetting or hopping.
    Each agent initiates at most $4 \cdot \tau$ interactions during $2\tau n$ \whp (see \cref{fact:moving}).
    Therefore, no agent increase their state by more than $4\tau$ before $t_a$ \whp, i.e., they remain in $I_{gather}$.
    
    On the other hand, there are less than $0.1n$ agents in $I_{gather}$.
    Thus, with probability at least $0.9$ per interaction, an agent $u \in I_{gather}$ resets.
    We will call this a \emph{success}.
    It is easy to see that \whp $u$ has at least one success in $t = \tau/2 \cdot n$ interactions.
    From \cref{lem:moving} with $\delta = 1$ it follows that $u$ initiates at least $\tau/4$ interactions in $t$ interactions with probability at least $1-n^{-2\cdot(c+4)}$.
    For any agent in $I_{gather}$, the probability that in $\tau/2 \cdot n$ interactions is at least $p=1-(1-0.9)^{\tau/4} = 1-n^{-3\cdot \log(10)\cdot (c+4)}$.
    We now extend this to all agents and $t_a$ interactions.
    Recall that the first agent in $L$ enters $I_{gather}$ at time $t_a < t'$.
    Agent $u$ resets at least once during each $\tau/2 \cdot n$ interactions during $t'$ interactions with probability at least
    $p^{2t'/(\tau \cdot n)} \geq 1-\LdauTheta{w} \cdot n^{-3\cdot \log(10) \cdot (c+4)}$.
    For sufficiently large $c$ that $u$ is the first $\tau$ states of $I_{gather}$ when the first agent of $L$ enters $I_{gather}$.
    A simple exchange argument yields that the probability for $u$ to not leave $I_{gather}$ before $t'$ is majorized by the probability of an agent $v \in L$ to not reach state $|Q| - \tau$ before $t'$.
    A union bound over all agents completes the proof.
\end{proof}

\section{Additional Details for the Adaptive Majority Protocol}
\label{apx:majority}

Finally, in this appendix we give the omitted details and full proofs for our adaptive majority protocol from \cref{sec:majority}.

\paragraph{Additional Details for the Pólya Subphase}

The main observation for this subphase is that we can model the opinion distribution after the Pólya Subphase by the Pólya-Eggenberger distribution.
Formally,
\begin{observation}
\label{obs:polya}
Assume the configuration at time $s_1-1$ is fixed and let $a = A_{s_1-1}$ and $b = B_{s_1-1}$.
Then $A_{e_1} \sim \PE(a, b, n-a-b)$ and $B_{e_1} \sim \PE(b, a, n - a - b)$, \whp.
\end{observation}
\begin{proof}
The observation follows from a coupling of the Pólya Subphase with the Pólya urn process.
Note that a similar observation has been previously used in \cite{DBLP:journals/corr/abs-2103-10366}.
Let $\ell_0 = a + b$ be the number of agents that have an opinion at time $s_1-1$.
In step $i$, the Pólya urn process picks an arbitrary undecided agent.
This agent chooses one of the $\ell_{i-1}$ agents that have an opinion uniformly at random and adopts its opinion, resulting in $\ell_{i} = \ell_{i-1}+1$.
It is now straightforward to couple the Pólya urn process with the Pólya Subphase: we simply discard all interactions that do not change the number of agents that have an opinion.

It remains to show that at time $e_1$ no undecided agents are left.
This follows from the result on the one-way epidemic \cref{lem:one-way-epidemic} (see also \cite{DBLP:journals/dc/AngluinAE08a}), together with the following observations:
By definition, all agents perform the Pólya Subphase during subphase $0$.
From the phase clock we get that there is at least an overlap of length $\LDAUTheta{n\cdot\log{n}}$ interactions where all agents are in the Pólya Subphase together. %TODO: exact value for theta?
This overlap is long enough for the one-way epidemic to conclude with probability $1-n^{-(c+2)}$.
\end{proof}

\Cref{obs:polya} now allows us to apply \cref{thm:polya_bound_total_balls} in order to prove concentration of $A_{e_1}$.
For convenience, the lemma is restated as follows.

\lempolya*

\begin{proof}
Recall that according to \cref{obs:polya} we have $A_{e_1} \sim \PE(a, b, n-a-b)$ and $B_{e_1} = n - A_{e_1}$, \whp, and $A_{e_1} + B_{e_1}= n$ \whp.
Let $\epsilon_p$ be the small constant from \cref{thm:polya_bound_total_balls} (see \cref{apx:auxiliary-results}).
We apply \cref{thm:polya_bound_total_balls} to $A_{e_1}$ and get for $\delta = \sqrt{((2+c)/\epsilon_p)\cdot \ln n}$ that
\begin{align*}
\Prob{A_{e_1} \leq \frac{n}{a+b} \cdot a \cdot \left( 1-\frac{\delta}{\sqrt{a}} \right) } &\leq 4 \cdot \exp\left(\epsilon_p \cdot \delta^2 \right) \leq 4 \cdot n^{-(c+2)}.
\end{align*}
In this case we have \whp that
\begin{equation}
\label{eq:polya-bound-application}
\begin{aligned}
 A_{e_1} - B_{e_1} 
& = 2 \cdot A_{e_1} - n \geq 2 \cdot \frac{n}{a+b} \cdot a \cdot \left( 1-\frac{\delta}{\sqrt{a}} \right) - n \\
& \geq n \cdot \left(\frac{2 + 2\cdot\beta}{2 + \beta} \cdot \left( 1-\frac{\delta}{\sqrt{a}} \right)-1\right) \\
&= n \cdot \left( \frac{\beta}{2 + \beta} \cdot \left( 1-\frac{\delta}{\sqrt{a}} \right) -\frac{\delta}{\sqrt{a}}\right) \\
& \geq n \cdot \left( \frac{\beta}{2 + \beta} - 2 \cdot \sqrt{\frac{ (2+ c) }{\epsilon_p \cdot \alpha \cdot \log e}} \right).
%TODO: Rechnung korrigieren, WTF?! log e?!?!?!
\end{aligned}
\end{equation}
Hence  $A_{e_1} - B_{e_1} = \LDAUOmega{n}$ for a sufficiently large constant $\alpha$ and given constant $\beta> 0$ with probability $1-n^{-(c+2)}$.
\end{proof}

\paragraph{Additional Details for the Cancellation Subphase}
Recall that in the cancellation subphase whenever an $A$ agent interacts with a $B$ agent, both become undecided.
Analogously to before, $s_2$ and $e_2$ are the first and the last time, respectively, when an agent performs an interaction in the Cancellation Subphase.
We now prove \cref{lem:cancellation}, which is restated here for convenience.

\lemcancellation*

\begin{proof}
Every agent with opinion $B$ can cancel out at most one other agent with opinion $A$, hence there are always at least $A_{s_2}-B_{s_2} = \LDAUOmega{n}$ many agents with opinion $A$ during the entire subphase.
Suppose an agent $u$ with opinion $B$ interacts with another agent $v$ chosen uniformly at random.
Then the probability that agent $v$ holds opinion $A$ is at least $\LDAUOmega{1}$.
All agents (including those that hold opinion $B$) initiate at least $ \LDAUTheta{\log n} $  interactions in the Cancellation Subphase.
Hence, for sufficiently large length of the subphase, every agent with opinion $B$ becomes undecided with probability at least $1-n^{-{c+3}}$.
The statement then follows via a union bound over all agents.
\end{proof}

\paragraph{Additional Details for the Broadcasting Subphase}

We prove the following lemma for the Broadcasting Subphase.
Recall that $s_3$ and $e_3$ are the first and the last time, respectively, when an agent performs an interaction in the Broadcasting Subphase.

\lembroadcasting*

\begin{proof}
The Broadcasting Subphase can be regarded as an epidemic spreading process, where the remaining opinion spreads to all other agents.
The proof follows immediately from \cref{lem:one-way-epidemic} along with the observation that there is a sufficiently long overlap in that subphase.
\end{proof}

\paragraph{Full Proof of \cref{pro:improved-bounds}}
We now give the full proof of \cref{pro:improved-bounds} which is restated here for convenience.
\proImprovedBounds*
\begin{proof}
In order to show \cref{pro:improved-bounds} we need a slightly more careful calculation in the proof of \cref{lem:polya} that gives us a better bound on $A_{e_1} - B_{e_1}$, the bias after the Pólya Subphase.

Consider \cref{eq:polya-bound-application} in the proof of \cref{lem:polya}.
In order to achieve the bias of $\LDAUTheta{n^{3/4 + \epsilon}}$ required by \cref{thm:undecided-dynamics}, we require for some constant $\epsilon > 0$ that
\begin{align*}
A_{1} - B_{1} &
%\stackrel{\eqref{eq:polya-bound-application}}{\geq} n \cdot \left( \frac{\beta}{2 + \beta} - 2 \cdot \sqrt{\frac{ 2 }{\epsilon_p \cdot c \cdot \log e}} \right) \geq n^{3/4 + \epsilon}.
\stackrel{\eqref{eq:polya-bound-application}}{\geq} n \cdot \left( \frac{\beta}{2 + \beta} - 2\sqrt{2} \cdot \left( \epsilon_p \cdot \alpha \cdot \log e \right)^{-1} \right) \geq n^{3/4 + \epsilon}.
\end{align*}
When we solve the expression in parentheses for $\alpha$ we obtain
\begin{align*}
\alpha &= \LDAUOmega{\beta^{-2}} \quad \text{ provided that } \quad \beta = \LDAUOmega{n^{-1/4+\epsilon}} . %\\
% \alpha &= \LDAUOmega{\left(\beta  - 1/3 \cdot n^{-1/4 + \epsilon} \right)^{-2}} \quad \text{ provided that } \quad \beta = \LDAUOmega{n^{-1/4+\epsilon}} . %\\
%\beta &= \Omega\left(c^{-1/2} + n^{-1/4+\epsilon} \right) \\
\end{align*}
which gives the claimed bounds on $\alpha$ and $\beta$, provided the constants in the asymptotic notation are large enough.
\end{proof}

\paragraph{Full Proof of \cref{pro:input-changes}}
Finally, we give the full proof of \cref{pro:input-changes} which is restated here for convenience.
\proInputChanges*
\begin{proof}
Recall that each phase of the clock consists of $\Theta(n \log n)$ interactions.
Hence $t_2 - t_1 = \Theta(n \log n)$.
Let $X_{[t_1,t_2]}$ be the random variable for input changes in $[t_1, t_2]$ and observe that $\Ex{X_{[t_1,t_2]}} = r \cdot (t_2 - t_1) = \ldauTheta{r \cdot n \log n}$.
We apply Chernoff bounds to $X_{[t_1,t_2]}$ and obtain for a sufficiently large constant $c'$ that $\Prob{X_{[t_1,t_2]} > c'\cdot r \cdot n\log n} \leq n^{-(c+1)}$.
This means that at most $ c'\cdot r \cdot n\log n$ agents change their input to the minority opinion in $[t_1, t_2]$ \whp.

We now distinguish two cases.

\paragraph{Case 1: $\beta \leq 1/(2c')$}
In the first case we have only a small bias.
We therefore have to revisit the Pólya Subphase once again.
Consider a modified process where all input changes that would originally occur during the Pólya Subphase take place before the Pólya Subphase starts.
We will now show via a coupling that handling input changes \emph{early} at the beginning of the phase does not alter the outcome of the protocol.
A straight-forward coupling shows that the number of agents that have the majority opinion in the modified process minorizes the same number in the original process.
(See \cite{JK77} for additional details on the Pólya-Eggenberger distribution).
In the modified process, we initially have $A_{s_1}' = A_{s_1} - c'\cdot r \cdot n\log n$ agents with opinion $A$.
Therefore, we have $A_{s_1}' \geq \alpha\log n - c' \cdot \beta \alpha \cdot \log n = \alpha \log n (1 - c' \beta)$ agents with opinion $A$.
If $\beta \leq 1/(2c')$ the statement immediately holds (provided $\alpha$ is sufficiently large) for $\alpha \geq \Omega(\beta^{-2})$.
A similar calculation shows that in this setting the (additive) bias drops only by a constant factor.
The statement then follows in this case analogously to the result without input changes.

\paragraph{Case 2: $\beta \geq 1/(2c')$}
In this case the additive bias $A_{s_1} - B_{s_1}$ is at least $\beta \cdot \alpha \cdot \log n$.
It follows that $A_{s_1}' - B_{s_1}' \geq \alpha \cdot \log n$.
This means that at the beginning of the Pólya Subphase we have at least $\ldauOmega{\log n}$ agents with opinion $A$ and a constant (multiplicative) bias towards $A$.
Again, the statement then follows in this case analogously to the result without input changes.

\medskip

In both cases the remaining bias and the number of agents with opinion $A$ is large enough such that the previous analysis (\cref{pro:adaptive-majority}) after the Pólya Subphase can be applied without without further modifications.
\end{proof}

\end{document}